\documentclass[12pt,a4paper]{article}
\usepackage[top=1in, bottom=1in, left=1in, right=1in]{geometry}
\usepackage{lineno}
\usepackage{longtable}
\usepackage[utf8]{inputenc}
\usepackage{amsmath}
\usepackage{amsthm}

\usepackage[T1]{fontenc}
\usepackage{cochineal}         
\usepackage[cochineal]{newtxmath}
\usepackage{microtype}
\usepackage{adjustbox}
\usepackage{xcolor}
\usepackage{graphicx}
\usepackage{float}
\usepackage[section]{placeins} 
\usepackage{microtype}      
\usepackage{flafter}        
\usepackage{etoolbox}       

\usepackage{verbatimbox}
\usepackage{booktabs}
\usepackage{algorithm}
\usepackage{algorithmic}
\usepackage{authblk}
\usepackage{lipsum}

\usepackage{caption} 
\usepackage{subcaption} 
\usepackage[T1]{fontenc}
\theoremstyle{definition}

\renewcommand{\arraystretch}{1.5}

\makeatletter
\def\thm@space@setup{\thm@preskip=1.2\parskip \thm@postskip=0pt}
\makeatother
\usepackage{tikz}
\usetikzlibrary{arrows}
\usepackage{xurl} 

\usepackage[numbers,sort&compress]{natbib}

\usepackage{hyperref}\hypersetup{
  colorlinks=true,
  citecolor=blue,
  linkcolor=red,
  filecolor=blue,
  urlcolor=blue,
}

\theoremstyle{plain}
\newtheorem{theorem}{Theorem}

\theoremstyle{definition}

\theoremstyle{remark}

\usepackage[title,titletoc]{appendix}

\usepackage{enumitem} 
\usepackage{mathtools} 
\usepackage{array}

\usepackage[section]{placeins}
\usepackage{authblk}

\newcommand{\etal}{\textit{et al}. }

\begin{document}
\pagenumbering{gobble} 

%

%
%
%
%
%
%

\clearpage 
\pagenumbering{arabic} 
\setcounter{page}{1} 

\title{A Predictive Model for Synergistic Oncolytic Virotherapy: Unveiling the Ping-Pong Mechanism and Optimal Timing of Combined Vesicular Stomatitis and Vaccinia Viruses}

\author{
    Joseph Malinzi$^{1,2,\dagger}$,  
    Amina Eladdadi$^{3,\ddagger}$, 
    Rachid Ouifki$^{4}$, 
    Raluca Eftimie$^{5}$, 
    Anotida Madzvamuse$^{6,7,8,9}$, 
    Helen M. Byrne$^{10}$ \\
    \small
    $^{1}$ Department of Mathematics, University of Eswatini, Kwaluseni M201, Eswatini\\
    $^{2}$ Institute of Systems Science, Durban University of Technology, Durban 4000, South Africa\\
    $^{3}$ Department of Mathematical Sciences, Rensselaer Polytechnic Institute, Troy, NY, USA\\
    $^{4}$ Department of Mathematics and Applied Mathematics, North-West University, Potchefstroom 2520, South Africa\\
    $^{5}$ Department of Mathematics, Université de Franche-Comté, France\\
    $^{6}$ Department of Mathematics, University of British Columbia, Vancouver, BC V6T 1Z2, Canada\\
    $^{7}$ Department of Mathematics and Computational Sciences, University of Zimbabwe, Harare, Zimbabwe\\
    $^{8}$ Department of Mathematics and Applied Mathematics, University of Pretoria, Pretoria 0132, South Africa\\
    $^{9}$ Department of Mathematics and Applied Mathematics, University of Johannesburg, PO Box 524 Auckland Park, Johannesburg 2006, South Africa\\
    $^{10}$ Wolfson Centre for Mathematical Biology, Mathematical Institute, University of Oxford, Oxford OX2 6GG, UK
}

\date{}
\maketitle

\renewcommand{\thefootnote}{\fnsymbol{footnote}}
\footnotetext[\value{footnote}]{$\dagger$Corresponding author: \texttt{josephmalinzi@aims.ac.za}}
\footnotetext[\value{footnote}]{$\ddagger$Corresponding author: \texttt{eladdadi@gmail.com}}
\renewcommand{\thefootnote}{\arabic{footnote}}

\vspace*{-1cm}
\begin{abstract}
\noindent We present a mathematical model that describes the synergistic mechanism of combined Vesicular Stomatitis Virus (VSV) and Vaccinia Virus (VV). The model captures the dynamic interplay between tumor cells, viral replication, and the interferon-mediated immune response, revealing a `ping-pong' synergy where VV-infected cells produce B18R protein that neutralizes interferon-$\alpha$, thereby enhancing VSV replication within the tumor. Numerical simulations demonstrate that this combination achieves complete tumor clearance in approximately 50 days, representing an 11\% acceleration compared to VV monotherapy (56 days), while VSV alone fails to eradicate tumors. Through bifurcation analysis, we identify critical thresholds for viral burst size and B18R inhibition, while sensitivity analysis highlights infection rates and burst sizes as the most influential parameters for treatment efficacy. Temporal optimization reveals that therapeutic outcomes are maximized through immediate VSV administration followed by delayed VV injection within a 1-19 day window, offering a strategic approach to overcome the timing and dosing challenges inherent in OVT.
\end{abstract}


\section{Introduction}
\label{sec:intro}

Oncolytic virotherapy (OVT) uses engineered viruses to selectively infect and lyse cancer cells. While showing promise as a monotherapy \cite{russell2022advances,ottolino2010intelligent,martin2018oncolytic,bartlett2013oncolytic,malinzi2021prospect}, its efficacy is often limited by antiviral immune responses and poor tumor penetration. Combination strategies, either with conventional therapies or with other oncolytic viruses (OVs), are therefore being actively explored to improve outcomes \cite{le2010synergistic,russell2022advances,malinzi2021prospect}. Pairing two complementary OVs can exploit distinct mechanisms of action, broaden tumor targeting, and create self-amplifying therapeutic effects while evading host immunity \cite{le2010synergistic,zhang2014combination}.

Experimental work by Le Boeuf \textit{et al.} \cite{le2010synergistic} demonstrated strong synergy between Vesicular Stomatitis Virus (VSV) and Vaccinia Virus (VV). VV establishes local infection pockets and expresses immune-evading genes, while VSV replicates rapidly and spreads efficiently. Together they produce a 'ping-pong' effect: VV-infected cells secrete the B18R protein, which neutralizes interferon-$\alpha$ (IFN-$\alpha$) and thereby enhances VSV replication \cite{le2010synergistic}. Other studies have validated similar synergistic combinations \cite{zhang2014combination,vaha2014overcoming,vaha2013resistance}.

The efficacy of this combination is rooted in the complementary biology of the two viruses. VSV is an RNA virus with rapid replication kinetics and high oncolytic specificity for interferon-defective cells, while VV is a large, tumor-selective DNA virus capable of extensive genetic modification and potent immune modulation. Their interaction creates a self-amplifying cycle: VV can infect tumor cells refractory to VSV and, through proteins like B18R, condition the tumor microenvironment to be non-responsive to antiviral cytokines such as IFN-$\alpha$. This dismantles local antiviral defenses, thereby liberating VSV to replicate and spread aggressively---an effect so potent that VSV spread can surpass VV even at lower inoculums \cite{le2010synergistic}. VSV, in turn, enhances VV propagation, leading to synergistic tumor cell killing \cite{le2010synergistic, hastie2012vesicular, barber2004vesicular, guse2011oncolytic, thorne2011immunotherapeutic}.

Beyond direct oncolysis, OVs stimulate antitumor immunity by remodeling the microenvironment and expressing immunomodulatory transgenes \cite{russell2022advances,chattopadhyay2024review}. Several OVs, including VV, VSV, and measles virus, have entered clinical trials, with some (e.g., Talimogene laherparepvec) achieving regulatory approval \cite{heo2013randomized,falls2016murine,karapanagiotou2012phase,galanis2010phase,perez2012design,kaufman2015oncolytic}. Despite these advances, challenges remain in optimizing delivery, dosing, and viral replication control \cite{malinzi2021prospect}.

Mathematical modeling has become an essential tool for understanding OV dynamics and designing effective combination protocols \cite{sherlock2023oncolytic,ramaj2023treatment,guo2023mathematical,pooladvand2022,li2022modeling,malinzi2021mathematical,heidbuechel2020mathematical,lee2020application,li2020hopf,jang2022mathematical,friedman2003analysis,wein2003validation, wu2001modeling,tao2005competitive,eftimie2018tumour,Malinzi14}. Recent models have examined OVs paired with enhancer viruses \cite{jenner2021silico}, CAR-T cells \cite{mahasa2022combination}, chemotherapy \cite{malinzi2018enhancement,malinzi2017modelling}, and checkpoint inhibitors \cite{friedman2018combination}. These studies underscore the critical importance of administration timing, dosing sequences, and parameter sensitivity in achieving synergistic outcomes.

This study addresses a key gap by developing a mathematical model that elucidates the synergistic dynamics of combined VV and VSV therapy. We aim to: (1) characterize the `ping-pong' interaction and its impact on tumor dynamics; (2) identify conditions for complete tumor eradication; (3) generate testable hypotheses for experimental validation; (4) optimize administration timing and dosage; and (5) pinpoint the most sensitive parameters for therapeutic optimization. 

The paper is organized as follows. Section~\ref{sec:ModelFramwork} presents the mathematical model. Section~\ref{sec:optimizingtaus} optimizes injection timing. Section~\ref{sec:math-analysis} provides stability, reproduction number, and bifurcation analyses. Section~\ref{sec:Simulations} shows numerical simulations. Section~\ref{subsec:reproduction-numbers} analyzes viral replication thresholds. Section~\ref{sec:results} synthesizes key findings, and Section~\ref{sec:conclusion} discusses implications and future directions.

\section{Mathematical Model Framework}
\label{sec:ModelFramwork}

We propose a mathematical model to investigate the synergistic anti-tumor effects of combining Vaccinia Virus (VV) and Vesicular Stomatitis Virus (VSV) \cite{le2010synergistic,vaha2014overcoming,zhang2014combination}. The model describes how the two oncolytic viruses and tumor cells interact within an avascular tumor environment following simultaneous injection of both viruses \cite{le2010synergistic}. The network of biological interactions is depicted in Figure~\ref{fig:schematic_diagram}, while model variables and parameters are summarized in Tables~\ref{tab:modelvars} and~\ref{tab:parametervalues1}, respectively.

\subsection{Model Assumptions}

The mathematical model incorporates several key biological mechanisms:
\paragraph{(1) Synergistic viral interactions and co-infection dynamics} 
Experimental evidence demonstrates that VSV exhibits limited replication capacity in certain cancer cell lines, such as HT29, when administered alone. However, when combined with VV, the latter virus creates infection portals in these otherwise resistant tumor populations, sensitizing them to subsequent VSV infection. This complementary viral tropism enables the virus pair to target a broader spectrum of tumor cells than either virus could target individually \cite{le2010synergistic}. The resulting co-infection dynamics produce synergistic tumor cell killing, where the combined effect substantially exceeds the additive effects of each virus administered separately \cite{le2010synergistic,zhang2014combination}. Importantly, this synergy does not require simultaneous infection of individual tumor cells by both viruses, but rather emerges through sequential viral actions that collectively transform the tumor microenvironment.

\paragraph{(2) Molecular mechanisms of synergy}
At the molecular level, synergy is mediated via VV's expression of the B18R protein, which functions as a potent interferon antagonist. This viral gene product binds to type I interferon molecules, inhibiting interferon-mediated signal transduction and neutralizing a key component of the innate antiviral defense system \cite{le2010synergistic,fritz2014recombinant}. This mechanism is particularly important since interferon-alpha (IFN-$\alpha$) normally exerts potent inhibitory effects on VSV replication by inducing interferon-stimulated genes that establish an antiviral state within cells \cite{le2010synergistic,kueck2019vesicular}. The reciprocal nature of this viral partnership is further evidenced by observations that VSV can be genetically modified to enhance VV spread, creating a mutually beneficial amplification cycle that some researchers have termed as a `ping-pong' effect \cite{le2010synergistic}.

\paragraph{(3) Quasi-steady state approximation for molecular kinetics}
The expression patterns of IFN-$\alpha$ and B18R occur on substantially faster timescales than the population dynamics of tumor cells and viruses, justifying the application of quasi-steady state approximations for these molecular components within the mathematical framework \cite{le2010synergistic}. The model also accounts for the different replication rates of the two viruses, with VSV replicating in approximately 6 hours compared to 10--12 hours for VV, along with differences in their burst sizes and clearance rates \cite{le2010synergistic,zhu2009growth}. Through this representation of coordinated viral kinetics, the model captures how VV, via its expression of B18R, conditions the tumor microenvironment to become non-responsive to antiviral cytokines, thereby enabling sustained VSV replication even in the presence of ongoing interferon signaling \cite{le2010synergistic,vaha2014overcoming}.

\begin{figure}[htbp]
    \centering
    \includegraphics[width=0.8\linewidth]{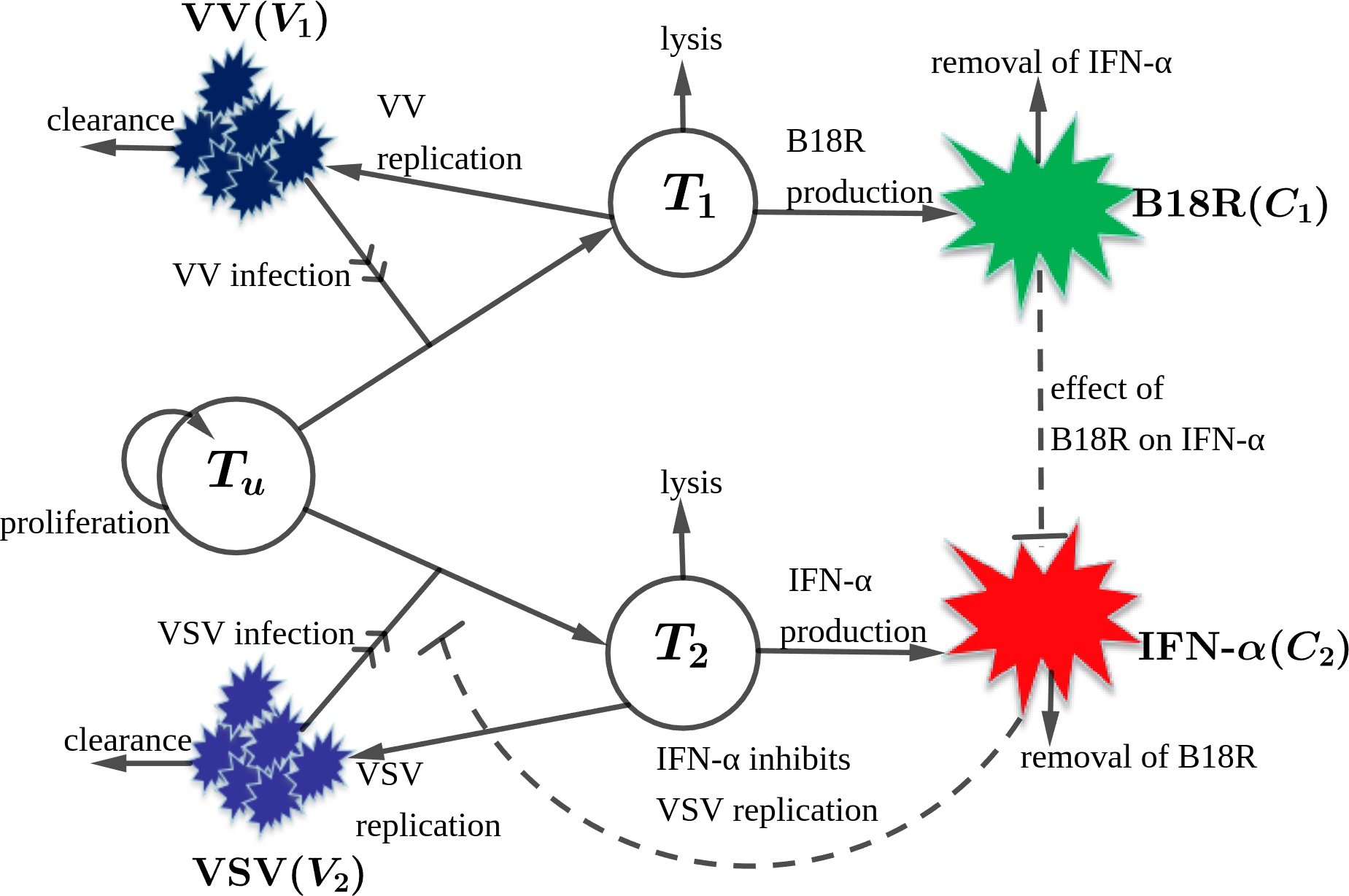}
	\caption{\small{\bf Schematic representation of the mathematical model describing the synergistic interaction between Vaccinia Virus (VV) and Vesicular Stomatitis Virus (VSV) with tumor cells.} The model consists of seven state variables: uninfected tumor cells ($T_u$), VV-infected tumor cells ($T_1$), VSV-infected tumor cells ($T_2$), free VV particles ($V_1$), free VSV particles ($V_2$), B18R protein concentration ($C_1$), and interferon-alpha concentration ($C_2$). \textbf{Solid lines} denote direct processes: viral infection, cell lysis, virus production, and molecular synthesis. \textbf{Dashed lines} represent inhibitory interactions. Key interactions include: proliferation of uninfected tumor cells ($T_u$) atm a rate $\alpha$ with carrying capacity $K$; infection of tumor cells by VV and VSV with Michaelis-Menten kinetics and infection rates $\beta_1$ and $\tilde{\beta}_2(C_2)$ respectively; lysis of infected tumor cells at rates $l_1$ and $l_2$; (4) virus production with burst sizes $b_1$ and $b_2$; virus clearance at rates $\gamma_1$ and $\gamma_2$; B18R production by VV-infected cells at rate $\alpha_1$ and B18R decay at rate $\mu_1$; IFN-$\alpha$ production by VSV-infected cells at rate $\alpha_2$ and IFN-$\alpha$ decay at rate $\mu_2$; inhibition of IFN-$\alpha$ by B18R at rate $\lambda$. The synergy arises as VV-infected cells produce B18R, which binds to and inhibits IFN-$\alpha$, thereby reducing the inhibitory effect on VSV infection described by $\tilde{\beta}_2(C_2) = \beta_2/(1 + C_2/C_2^*)$.}
    \label{fig:schematic_diagram}
\end{figure}

\subsection{Model Equations}

The mathematical model comprises seven coupled ordinary differential equations describing the temporal dynamics of the tumor cells, oncolytic viruses, and associated molecular mediators (Table~\ref{tab:modelvars}). The system tracks uninfected tumor cells ($T_u$), VV-infected cells ($T_1$), VSV-infected cells ($T_2$), free viruses ($V_1$, $V_2$), and molecular regulators B18R ($C_1$) and IFN-$\alpha$ ($C_2$) that mediate viral synergy. The complete system is described by:

\begin{align}
\frac{dT_u}{dt} & = \underbrace{\alpha T_u\left(1 - \frac{T_u + T_1+T_2}{K}\right)}_{\text{Tumor proliferation}} - \underbrace{\beta_1\frac{T_uV_1}{K_1 + T_u + T_1+T_2}}_{\text{Infection by VV}} - \underbrace{\tilde{\beta}_2(C_2)\frac{T_uV_2}{K_2 + T_u + T_1+T_2}}_{\substack{\text{IFN-}\alpha\text{-modulated} \\ \text{infection by VSV}}}, \label{eq:uninfected} \\
\frac{dT_1}{dt} & = \underbrace{\beta_1\frac{T_uV_1}{K_1 + T_u + T_1+T_2}}_{\text{Infection by VV}} - \underbrace{l_1T_1}_{\text{Lysis}}, \label{eq:infected1} \\
\frac{dT_2}{dt} & = \underbrace{\tilde{\beta}_2(C_2)\frac{T_uV_2}{K_2 + T_u + T_1+T_2}}_{\substack{\text{IFN-}\alpha\text{-modulated} \\ \text{infection by VSV}}} - \underbrace{l_2 T_2}_{\text{Lysis}}, \label{eq:infected2} \\
\frac{dV_1}{dt} & = \underbrace{b_1 l_1 T_1}_{\substack{\text{Virus release} \\ \text{from VV-infected cell lysis}}} - \underbrace{\beta_1\frac{T_uV_1}{K_1 + T_u + T_1+T_2}}_{\substack{\text{Loss due to} \\ \text{VV infection}}} - \underbrace{\gamma_1 V_1}_{\text{VV clearance}}\hspace*{-0.35cm}, \label{eq:V1} \\
\frac{dV_2}{dt} & = \underbrace{b_2 l_2 T_2}_{\substack{\text{Virus release} \\ \text{from VSV-infected cell lysis}}} - \underbrace{\tilde{\beta}_2 (C_2)\frac{T_uV_2}{K_2 + T_u + T_1+T_2}}_{\substack{\text{Loss due to} \\ \text{VSV infection}}} - \underbrace{\gamma_2 V_2}_{\text{VSV clearance}}\hspace*{-0.35cm}, \label{eq:V2} \\
\frac{dC_1}{dt} & = \underbrace{\alpha_1 T_1}_{\substack{\text{Production of B18R} \\ \text{by VV-infected cells}}} - \underbrace{\mu_1 C_1}_{\substack{\text{Natural decay} \\ \text{of B18R}}}\hspace*{-0.35cm}, \label{eq:C1} \\
\frac{dC_2}{dt} & = \underbrace{\alpha_2 T_2}_{\substack{\text{Production of IFN-}\alpha \\ \text{by VSV-infected cells}}} - \underbrace{\mu_2 C_2}_{\substack{\text{Natural decay} \\ \text{of IFN-}\alpha}} - \underbrace{\lambda C_1 C_2}_{\substack{\text{Neutralization of IFN-}\alpha \\ \text{by B18R}}}\hspace*{-0.8cm}, \label{eq:C2}
\end{align}
with initial conditions:
\begin{equation}
\begin{rcases}
T_u(0) = T_{u_{0}},\ T_1(0) = 0,\ T_2(0) = 0,\ V_1(0) = V_{1_{0}},\ V_{2}(0) = V_{2_{0}},\ C_1(0) = 0,\ C_2(0) = 0 \\
\end{rcases}
\label{eq:ICs}
\end{equation}
We propose 
\begin{align}
\tilde{\beta}_2(C_2) = \dfrac{\beta_2}{1 + \frac{C_2}{C_2^*}}
\label{eq:Synergy_function}
\end{align}
to model the IFN-$\alpha$-modulated infection rate of VSV. This expression captures the saturating inhibitory effect of IFN-$\alpha$ on viral infectivity: $\beta_2$ represents the maximal infection rate in a fully permissive environment ($C_2 = 0$), while the half-saturation constant $C_2^*$ quantifies VSV's sensitivity to the cytokine. The hyperbolic decrease with increasing $C_2$ is mechanistically grounded in interferon signaling biology, where IFN-$\alpha$ binding activates JAK/STAT pathways to establish an antiviral state, thereby reducing viral entry and replication in a concentration-dependent manner \cite{samuel2001antiviral}. This formulation can be justified by established models of ligand-receptor kinetics in virology \cite{padmanabhan2014mathematical} and directly represents the experimental observation that VV's B18R protein enhances VSV spread by neutralizing IFN-$\alpha$ \cite{le2010synergistic}.

It is worth noting that the following properties are assumed, based on \cite{le2010synergistic}: $\beta_2 \ll \beta_1$ when $C_2=0$, reflecting VV's higher initial infectivity; $l_1 \ll l_2$, representing VV's characteristically slower replication and lysis cycle compared to the rapid lytic cycle of VSV; and $\alpha_1, \mu_1 \gg 1$, ensuring that the dynamics of the B18R protein occur on a fast timescale, which supports the use of a quasi-steady state approximation. This parameterization establishes the baseline for the complementary synergy, where VV acts as the immune-modulating pioneer, enabling the potent oncolytic activity of VSV.

Model parameters (Table~\ref{tab:parametervalues1}) were calibrated against experimental data from~\cite{le2010synergistic} using nonlinear least squares optimization \cite{bates1988nonlinear}. Data extracted via WebPlotDigitizer \cite{Rohatgi2023} from HT29 and 4T1 tumor volume measurements under PBS control, VSV monotherapy, VV monotherapy, and combination therapy (Tables~\ref{tab:HT29} and \ref{tab:4T1}) were used for parameter estimation. The calibrated model accurately captures the characteristic synergistic tumor reduction (Figure~\ref{fig:modelfitting}), providing a validated foundation for predictive simulations.

\begin{table}[htbp]
\centering
\caption{\textbf{Model Variables}}
\label{tab:modelvars}
\begin{tabular}{@{}lll@{}}
\hline
Variable & Description & Units \\ 
\hline
$T_u(t)$ & Uninfected tumor density & Cells per mm$^3$ \\
$T_1(t)$ & VV-infected tumor cell density & Cells per mm$^3$ \\
$T_2(t)$ & VSV-infected tumor cell density & Cells per mm$^3$ \\
$V_1(t)$ & Free Vaccinia Virus (VV) & PFU\\
$V_2(t)$ & Free Vesicular Stomatitis Virus (VSV) & PFU\\
$C_1(t)$ & B18R gene product & Molecules per mm$^3$ \\
$C_2(t)$ & Type I interferon (IFN-$\alpha$) & Molecules per mm$^3$ \\
\hline
\end{tabular}
\end{table}

\begin{table}[htbp]
\centering
\caption{\textbf{Model Parameters, Descriptions, and Baseline Values}}
\label{tab:parametervalues1}
\begin{tabular}{@{}>{\raggedright}p{1.2cm}>{\raggedright}p{4cm}lll@{}}
\toprule
\textbf{Symbol} & \textbf{Description} & \textbf{Value} & \textbf{Units} & \textbf{Reference} \\ 
\midrule

\multicolumn{5}{l}{\textbf{Tumor Growth Parameters}} \\
\cmidrule(r){1-5}
$\alpha$ & Tumor growth rate & 0.2--0.9 & day$^{-1}$ & Estimated \\   
$K$ & Tumor carrying capacity & $10^8$--$10^9$ & cells/mm$^3$ & Estimated \\
\addlinespace

\multicolumn{5}{l}{\textbf{Viral Infection Parameters}} \\
\cmidrule(r){1-5}
$\beta_i$ & Tumor cell infection rates & 0.001--0.1 & cells.virus$^{-1}$.day$^{-1}$ & \cite{ZTK08} \\
$K_i$ & Michaelis-Menten constants & $10^8$ & cells/mm$^3$ & \cite{kirschner1998modeling} \\ 
\addlinespace

\multicolumn{5}{l}{\textbf{Infected Cell and Virus Dynamics}} \\
\cmidrule(r){1-5}
$l_i$ & Infected tumor cell death rates & 0.125--1 & day$^{-1}$ & \cite{hwang2013engineering,zeh2015first} \\
$\gamma_i$ & Virus decay rates & 0.01--0.1 & day$^{-1}$ & \cite{ZTK08} \\
$b_i$ & Virus burst sizes & 2--8000 & viruses/cell & \cite{zhu2009growth,rager1982role,BD90,almuallem2021oncolytic} \\
\addlinespace

\multicolumn{5}{l}{\textbf{B18R Protein Dynamics}} \\
\cmidrule(r){1-5}
$\alpha_1$ & Production rate of B18R & 0.1--1 & molecules cell$^{-1}$ day$^{-1}$ & Derived \\
$\mu_1$ & Decay rate of B18R & 0.001--0.5 & day$^{-1}$ & Derived \\  
\addlinespace

\multicolumn{5}{l}{\textbf{IFN-$\alpha$ Dynamics}} \\
\cmidrule(r){1-5}
$\alpha_2$ & IFN-$\alpha$ production rate & 0.1--1 & molecules cell$^{-1}$ day$^{-1}$ & Derived \\
$\mu_2$ & Decay rate of IFN-$\alpha$ & 0.001--0.5 & day$^{-1}$ & Derived \\ 
$\lambda$ & IFN-$\alpha$ inhibition rate by B18R & 0.01--1 & day$^{-1}$ & Derived \\ 

\bottomrule
\addlinespace
\multicolumn{5}{@{}l}{\footnotesize \textit{Note:} $i = 1,2$ for virus-specific parameters (VV: $i=1$, VSV: $i=2$)} \\
\end{tabular}
\end{table}

\subsection{Biological Mechanisms}

\paragraph{Tumor Cell Dynamics}
Tumor growth follows logistic dynamics \cite{ZTK08,bajzer1997mathematical,komarova2010ode,malinzi2018enhancement,malinzi2017modelling}, capturing the macroscopic expansion of tumor volume while accounting for complex microscopic processes including proliferation regulation and immune escape \cite{Benzekry2014}. Equations (\ref{eq:uninfected}), (\ref{eq:infected1}), and (\ref{eq:infected2}) describe the dynamics of the uninfected tumor cells ($T_u$) and virus-infected populations ($T_1$ for VV, $T_2$ for VSV). The logistic growth term $\alpha T_u\left(1 - \frac{T_u + T_1 + T_2}{K}\right)$ incorporates intrinsic growth rate $\alpha$ and carrying capacity $K$, representing resource-limited expansion. Viral infection follows Michaelis-Menten kinetics of the form $\beta_i T_u V_i/(K_i + T_u + T_1 + T_2)$, where $i=1,2$, to account for saturation effects \cite{komarova2010ode,wagner1973properties,yu2023exploring,de2006modeling,agrawal2014optimal,malinzi2018enhancement}. This formalism models the saturable, competitive nature of viral infection, where the sum $T_u + T_1 + T_2$ in the denominator accounts for competition between all cells (uninfected and infected) for virus binding, as they present similar surface receptors. The functional form $\frac{V}{K + T_{total}}$ captures the biological reality where infection rates plateau at high cell densities, as free virus particles become the limiting resource. Infection of uninfected tumor cells increases the number of infected tumor cells and decreases the number of viruses. The infected tumor cells, $T_{1}$ and $T_{2}$, have a lifespan $1/l_1$ and $1/l_2$ respectively.

\paragraph{Viral Replication and Synergy}
The dynamics of vaccinia virus (VV) and vesicular stomatitis virus (VSV) follow established mathematical formulations for oncolytic virus behavior \cite{komarova2010ode,bajzer1997mathematical,yu2023exploring,de2006modeling,agrawal2014optimal,malinzi2018enhancement}. Equations (\ref{eq:V1}) and (\ref{eq:V2}) describe the temporal evolution of free VV ($V_1$) and VSV ($V_2$) populations, respectively. Viral replication occurs at a rate $l_i b_i T_i$, where $b_i$ is the burst size, defined as the number of new viruses released upon the lysis of an infected cell. Concurrently, the depletion of free viruses due to infection of uninfected tumor cells is modeled by the terms $\beta_i T_u V_i/(K_u + T_u + T_I)$, which employ Michaelis-Menten kinetics to account for saturation effects. Both virus populations undergo natural clearance from the tissue environment at rates $\gamma_i$, and are supplied via controlled administration protocols.
The synergistic interaction between VV and VSV is mediated by the IFN-$\alpha$-dependent VSV infection rate, $\tilde{\beta}_2(C_2)$, defined in Eq. \eqref{eq:Synergy_function}. It is important to note that we consider $\beta_1$ to be constant, reflecting the assumption that VV's intrinsic infectivity is largely independent of local IFN-$\alpha$ due to its immune-evasion genes \cite{le2010synergistic}. A full system-level `ping-pong' effect, where VSV activity in turn benefits VV spread, can nonetheless emerge indirectly through the model dynamics. This can occur via mechanisms such as VSV-mediated reduction of the shared target cell pool, which alters the competitive landscape for VV.
 	
\paragraph{Gene Product and Interferon Dynamics}
The molecular mediators central to the synergistic mechanism are described by equations (\ref{eq:C1}) and (\ref{eq:C2}), which govern the dynamics of the vaccinia gene product B18R ($C_1$) and IFN-$\alpha$ ($C_2$). The B18R protein is produced by VV-infected tumor cells ($T_1$) at a rate which is proportional to $T_1$ with rate constant $\alpha_1$. This viral protein undergoes natural decay at a rate $\mu_1$. IFN-$\alpha$ is a key component of the innate antiviral response, which is produced by VSV-infected cells ($T_2$) at rate $\alpha_2$. Interferons induce the expression of interferon-stimulated genes (ISGs), many of which establish an antiviral state that inhibits viral replication \cite{kueck2019vesicular,fritz2014recombinant}. Although the precise mechanisms of action for many interferon-induced antiviral proteins remain incompletely characterized \cite{kueck2019vesicular}, IFN-$\alpha$ is known to inhibit VSV replication \cite{le2010synergistic}. The B18R protein antagonizes this antiviral effect represented by the term $\lambda C_1 C_2$ in Equation \eqref{eq:C2}, where $\lambda$ quantifies the inhibition rate. This molecular interaction creates the fundamental synergy mechanism: B18R-mediated reduction in IFN-$\alpha$ enhances VSV infection via the synergistic function $\tilde{\beta}_2(C_2)$, establishing a feed-forward loop that amplifies therapeutic efficacy.

\section{Optimization of Virus Injection Times}
\label{sec:optimizingtaus}

Building upon the established synergistic mechanisms between VV and VSV, we investigate the critical question of optimal administration timing to maximize therapeutic efficacy. The temporal coordination of viral delivery represents a crucial clinical consideration, as the sequence and timing of combination therapies can significantly impact treatment outcomes. This analysis investigates the challenge of determining whether simultaneous or sequential administration yields superior tumor control, and if sequential, what time delay optimizes the synergistic interaction.

\subsection{Modeling Delayed Injections}

To model the temporal aspects of virus administration, we incorporate time-delayed injection events into the viral dynamics equations. The instantaneous delivery of viral doses at time points $\tau_1$ (for VV) and $\tau_2$ (for VSV) is represented using Dirac delta functions $\delta(t - \tau_i)$. The modified virus equations become:

\begin{align}
\frac{dV_1}{dt} &= b_1 l_1 T_1 - \beta_1 \frac{T_u V_1}{K_1 + T_u + T_1 + T_2} - \gamma_1 V_1 + V_{1_0} \delta(t - \tau_1), \label{eq:V1_mod} \\
\frac{dV_2}{dt} &= b_2 l_2 T_2 - \tilde{\beta}_2 (C_2) \frac{T_u V_2}{K_2 + T_u + T_1 + T_2} - \gamma_2 V_2 + V_{2_0} \delta(t - \tau_2), \label{eq:V2_mod}
\end{align}
where $V_{1_0}$ and $V_{2_0}$ represent the viral loads administered at times $\tau_1$ and $\tau_2$, respectively. For numerical implementation, the delta functions are approximated by Gaussian distributions, 
\[\delta\left(t-\tau_i\right) \approx \frac{1}{\sigma_i \sqrt{2\pi}} \exp \left(-\frac{(t - \tau_i)^2}{2\sigma_i^2}\right), i=1,2,\]
where \(\sigma_i\) determines the spread of the Gaussian approximation and reflects the temporal uncertainty or duration of the virus injection at time \(\tau_i\). A smaller \(\sigma_i\) corresponds to a more acute, closer-to-instantaneous injection, while a larger \(\sigma_i\) smooths the injection over a wider time interval.

\subsection{Optimal Administration Strategy}

The optimization analysis reveals a clear pattern in the temporal coordination of VV and VSV administration. Figure \ref{fig:combined_timing} demonstrates that therapeutic outcomes are maximized through a specific sequential strategy: immediate administration of VSV ($\tau_2 = 0$) followed by delayed delivery of VV. Panel (a) shows that the optimal timing for VV administration ($\tau_1$) spans a broad window between 1 and 19 days post-initial treatment, providing clinical flexibility in scheduling. Conversely, panel (b) illustrates that delaying VSV administration progressively diminishes treatment efficacy, as evidenced by the monotonic increase in tumor concentration with increasing $\tau_2$.

Figure \ref{fig:varying_tau} provides additional validation of this strategy through comprehensive trajectory analysis. Panel (a) demonstrates that specific positive values of $\tau_1$ consistently yield lower tumor concentrations across the treatment timeline, while panel (b) confirms that immediate VSV administration ($\tau_2 = 0$) remains the most effective approach regardless of VV timing. This robust pattern emerges from the fundamental asymmetry in the viral interaction: VSV requires the interferon-antagonizing environment created by VV, but VV's therapeutic benefit depends on the presence of established VSV infection to maximize its synergistic potential.

\begin{figure}[htbp]
    \centering
    \begin{subfigure}{0.48\textwidth}
        \centering
        \includegraphics[width=\textwidth]{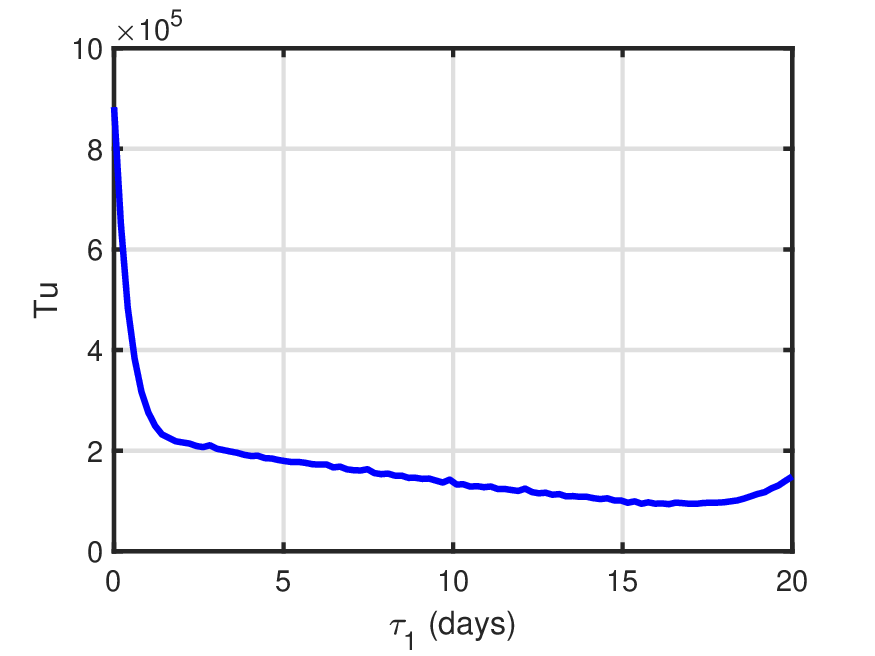}
        \caption{VV timing ($\tau_1$)}
        \label{fig:tau1-var}
    \end{subfigure}
    \hfill
    \begin{subfigure}{0.48\textwidth}
        \centering
        \includegraphics[width=\textwidth]{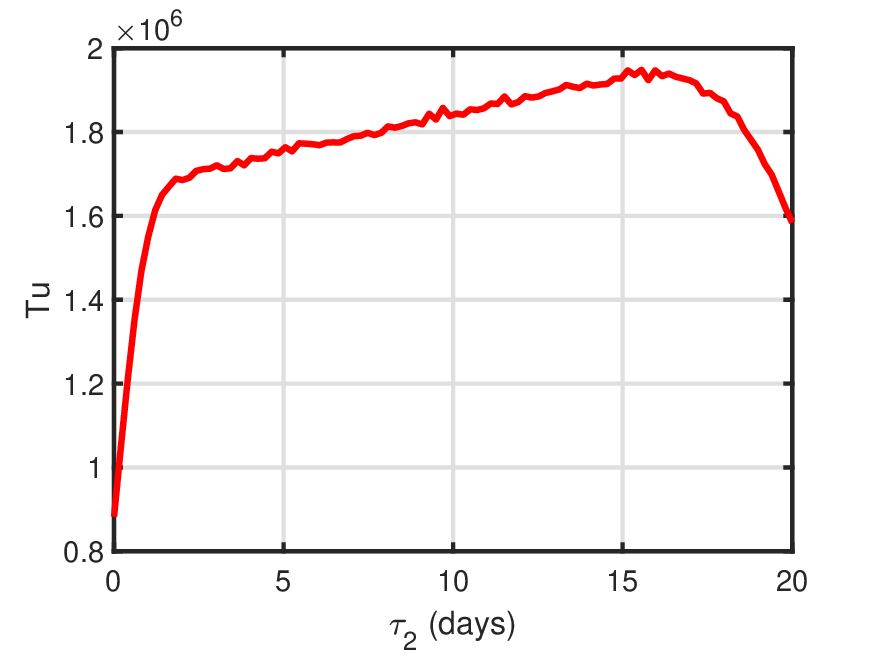}
        \caption{VSV timing ($\tau_2$)}
        \label{fig:tau2-var}
    \end{subfigure}    
    \caption{\small\textbf{Optimization of virus injection timing for VV-VSV combination therapy.} 
    (a) Final tumor burden as a function of VV injection time $\tau_1$ shows optimal therapeutic window between 1-19 days when VSV is administered immediately ($\tau_2 = 0$). 
    (b) Final tumor burden increases monotonically with VSV delay $\tau_2$, indicating immediate VSV administration is optimal. 
    The combined results establish the optimal strategy: immediate VSV followed by delayed VV injection (1-19 days).}
    \label{fig:combined_timing}
\end{figure}

\begin{figure}[htbp]
    \centering
    \begin{subfigure}{0.48\textwidth}
        \centering
        \includegraphics[width=\textwidth]{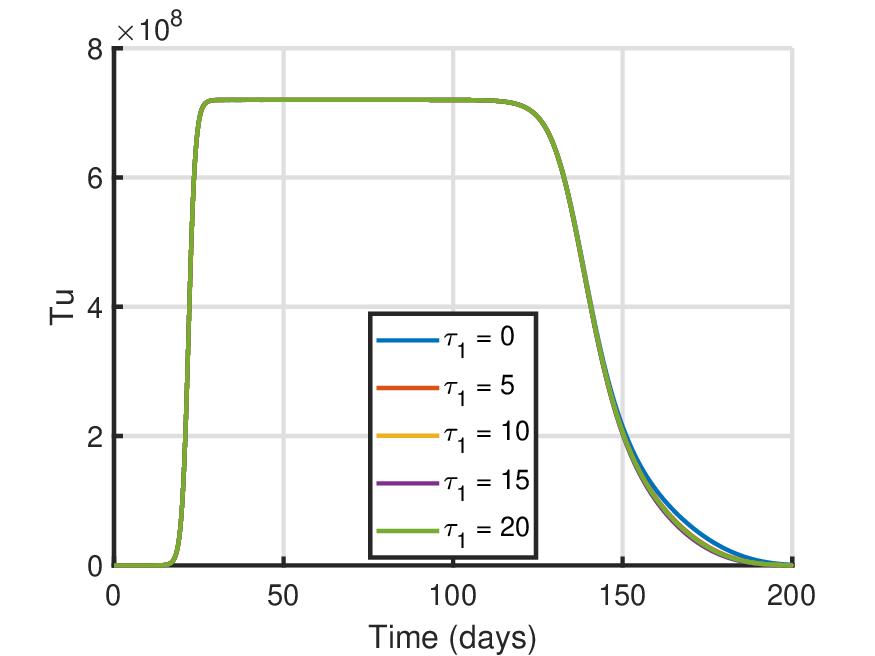}
        \caption{VV timing ($\tau_1$)}
        \label{fig:varying-tau1}
    \end{subfigure}
    \hfill
    \begin{subfigure}{0.48\textwidth}
        \centering
        \includegraphics[width=\textwidth]{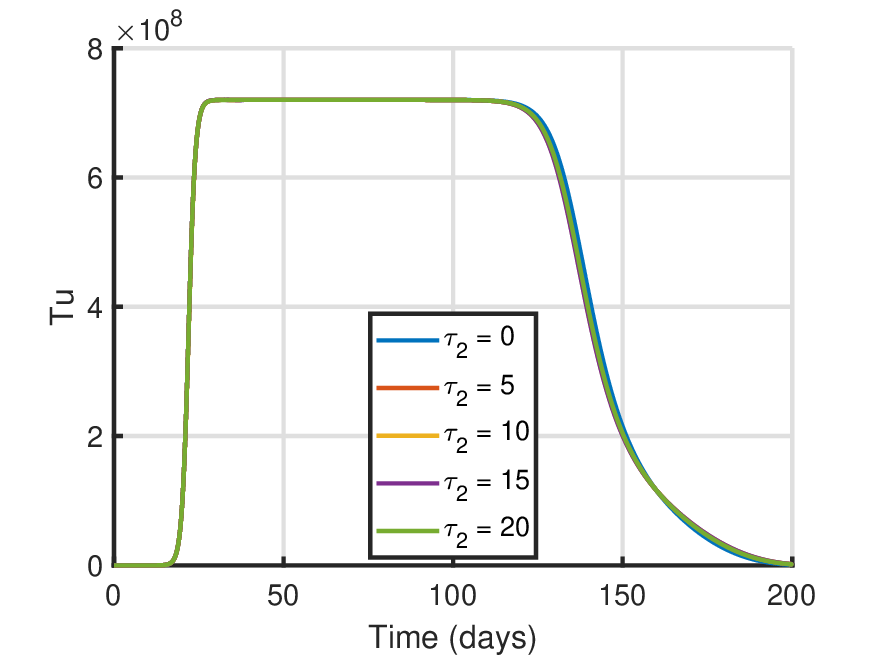}
        \caption{Temporal dynamics}
        \label{fig:varying-tau2}
    \end{subfigure}
    
    \caption{\small\textbf{Optimization of virus injection timing for VV-VSV combination therapy.} 
    (a) Temporal tumor dynamics demonstrate that immediate VSV injection ($\tau_2 = 0$) yields lower tumor concentrations across treatment timeline. 
    (b) Temporal tumor dynamics demonstrate that delayed VV injection ($\tau_1 > 0$) yields lower tumor concentrations across treatment timeline. 
    Collectively, these analyses establish the optimal strategy: immediate VSV injection followed by delayed VV administration within 1-19 days.}
    \label{fig:varying_tau}
\end{figure}

\section{Mathematical Analysis}
\label{sec:math-analysis}

The analysis begins by establishing fundamental well-posedness properties through Theorem~\ref{thm:well-posedness}, which confirms that solutions exist uniquely and remain confined within biologically realistic bounds which are essential prerequisites for meaningful biological interpretation and reliable numerical simulation.

To enhance analytical tractability and identify the fundamental dimensionless groups governing system behavior, we non-dimensionalize the system (complete derivations in Appendix \ref{app:non-dim}). This dimensional reduction reveals the key parameter combinations that dictate therapeutic outcomes, providing deeper insight into the system's intrinsic structure.

Building upon this foundation, we leverage the disparate timescales between molecular signaling and cellular dynamics by applying a quasi-steady state approximation to the B18R and IFN-$\alpha$ subsystems. As rigorously justified in Appendix \ref{app:qssa}, the rapid kinetics of these molecular components relative to tumor-virus interactions validate this reduction, which significantly simplifies the system while preserving its essential synergistic dynamics and enabling more efficient computational analysis.

The analytical foundation of our study centers on a comprehensive equilibrium analysis to characterize the long-term therapeutic landscape. Theorems~\ref{thm:stability12}, \ref{thm:stability22}, and \ref{thm:stability32} systematically establish the stability properties of both monotherapy and combination equilibria, delineating the mathematical conditions that enable viral persistence to sustain tumor control. These theoretical insights are further validated through numerical analysis using experimentally calibrated parameters, with a complete enumeration of biologically feasible steady states and their corresponding stability properties detailed in Table~\ref{tab:steady_states}. 

\subsection{Well-posedness}
\label{subsec:well-posedness}

\begin{theorem}[Well-posedness]\label{thm:well-posedness}
The dimensionless model equations~\eqref{eq:yu}-\eqref{eq:x2} satisfies the following mathematical properties:
\begin{enumerate}[label=(\alph*)]
	\item Existence and Uniqueness: There exists a unique solution to the model in the domain\\ $(y_u, y_1, y_2, x_1, x_2) \in \mathbb{R}^5_+$ on the maximal interval $[0, t_{\text{max}}]$ with $t_{\text{max}} >0$.
	\item Positivity: If $y_u(0) \geq 0$, $y_1(0) \geq 0$, $y_2(0) \geq 0$, $x_1(0) > 0$, and $x_2(0) > 0$, then the solutions satisfy $y_u(t) \geq 0$, $y_1(t) \geq 0$, $y_2(t) 	\geq 0$, $x_1(t) > 0$, and $x_2(t) > 0$ for all $t\in[0, t_{\text{max}}]$.
	\item Boundedness: The domain
	\[
		\Omega = \left\{(y_u, y_1, y_2, x_1, x_2) \in \mathbb{R}^5_+ ; y_u + y_1 + y_2 \leq 1, x_1(t) \leq \frac{b_1}{\delta_1}, x_2(t) \leq \frac{b_2}{\delta_2}\right\}
	\]
is positively invariant, ensuring all solutions remain bounded within biologically meaningful limits.
\end{enumerate}
\end{theorem}
The proof is provided in Appendix \ref{app:proofs-well-posedness}.

\paragraph{Biological Interpretation}
Existence and uniqueness guarantee deterministic tumor-virus dynamics, producing consistent predictions from given initial conditions. Positivity maintains non-negative population densities, preventing biologically impossible states. Most crucially, boundedness captures essential physical constraints: tumor growth remains limited by spatial constraints and resource availability, while viral populations are regulated by production rates and clearance mechanisms. 

\subsection{Qualitative Analysis of Equilibrium States} 
\label{subsec:Qualitative_analysis}

We now investigate the long term behavior of the model solutions for the full model \eqref{eq:yu}-\eqref{eq:x2}.

\subsubsection{VV monotherapy dynamics}

\begin{theorem}  \label{thm:stability12}
Without VSV, there are no VSV-infected tumor cells ($y_2=0$) and no IFN-$\alpha$ is produced ($z_2=0$). The VV-only model has three equilibrium points:
\begin{enumerate}
    \item $X^1_0 = (y_u, y_1, x_1) = (0,0,0)$ which is unstable,
    \item $X^1_1 = (y_u, y_1, x_1) = (1,0,0)$ which is stable if $a_1(\delta_1+\beta_1+\delta_1k_1) > b_1\beta_1$,
    \item A co-existence equilibrium point $X^1_c = (y_u^*, y_1^*, x_1^*)$ where $y_u^*, y_1^*, x_1^* > 0$, which exists when $a_1(\delta_1+\beta_1+\delta_1\kappa_1) < b_1\beta_1$ and is locally asymptotically stable when it exists.
\end{enumerate}
Detailed proof is in Appendix \ref{app:VV_only_model}. 
\end{theorem}

\paragraph{Biological Interpretation}
The mathematical analysis reveals why VV monotherapy shows limited efficacy: complete tumor eradication is unstable, while partial control through viral persistence represents the most likely outcome. The stability condition $a_1(\delta_1+\beta_1+\delta_1\kappa_1)>b_1\beta_1$ shows treatment failure occurs when viral burst size ($b_1$) is insufficient relative to viral decay ($\delta_1$), infection efficiency ($\beta_1$), and saturation effects ($\kappa_1$).

\subsubsection{VSV monotherapy dynamics} \label{subsec:VSV_only_model}

\begin{theorem} \label{thm:stability22}
In the absence of Vaccinia Virus (VV), there are no VV-infected tumor cells ($y_1 = 0$) and no B18R protein is produced ($z_1 = 0$). The VSV-only model has three biologically relevant equilibrium points:
\begin{enumerate}
    \item The tumor-free state $X^2_0 = (y_u, y_2, x_2) = (0, 0, 0)$, which is unstable.
    \item The tumor-dominant state $X^2_1 = (y_u, y_2, x_2) = (1, 0, 0)$, which is stable when $a_2(\delta_2 + \beta_2 + \delta_2\kappa_2) > b_2\beta_2$.
    \item  A co-existence equilibrium $X^2_c = (y_u^*, y_2^*, x_2^*)$ where $y_u^*, y_2^*, x_2^* > 0$, which exists when $a_2(\delta_2+\beta_2+\delta_2\kappa_2) < b_2\beta_2$ and is locally asymptotically stable when it exists.
\end{enumerate}
\end{theorem}
Proof is in Appendix \ref{app:VSV_only_model}. 

\paragraph{Biological Interpretation}
The stability analysis reveals that VSV monotherapy has limitations: complete tumor eradication ($X^2_0$) is unstable, while tumor dominance ($X^2_1$) becomes stable when viral replication capacity is insufficient. The condition $a_2(\delta_2 + \beta_2 + \delta_2\kappa_2) > b_2\beta_2$ indicates treatment failure occurs when the viral burst size ($b_2$) is too small relative to viral decay rate ($\delta_2$), infection efficiency ($\beta_2$), and saturation effects ($\kappa_2$). 

\subsubsection{Combined VV-VSV therapy dynamics}
\label{subsec:VV_VSV_model}

\begin{theorem} \label{thm:stability32}
The complete VV-VSV model described by equations \eqref{eq:yu}-\eqref{eq:x2} exhibits at least five fundamental equilibrium states:
\begin{enumerate}
    \item Tumor-free state $X^3_0 = (y_u, y_1, y_2, x_1, x_2) = (0,0,0,0,0)$, which is mathematically unstable
    \item Tumor-dominant state $X^3_1 = (1,0,0,0,0)$, stable when both $a_1(\delta_1+\beta_1+\delta_1\kappa_1) > b_1\beta_1$ and $a_2(\delta_2+\beta_2+\delta_2\kappa_2) > b_2\beta_2$
    \item VV-only coexistence state $X^1_c = (y_u^*, y_1^*, 0, x_1^*, 0)$, stable when $a_1(\delta_1+\beta_1+\delta_1\kappa_1) < b_1\beta_1$
    \item VSV-only coexistence state $X^2_c = (y_u^*, 0, y_2^*, 0, x_2^*)$, stable when $a_2(\delta_2+\beta_2+\delta_2\kappa_2) < b_2\beta_2$
    \item Full co-existence equilibria $X^3_c$ representing both viruses persisting with partial tumor control, which exist when both single-virus coexistence states are unstable and are determined by solving the full system of algebraic equations obtained by setting the left-hand sides of \eqref{eq:yu}-\eqref{eq:x2} to zero.
\end{enumerate}
\end{theorem}

Proof is in Appendix \ref{app:VV-VSV-Combined-Model}. 

\paragraph{Biological Interpretation}
The mathematical instability of the tumor-free state ($X^3_0$) reflects the biological challenge of achieving complete tumor eradication in complex tumor ecosystems, where residual microenvironments often permit tumor persistence. Similarly, the tumor-dominant state ($X^3_1$) only stabilizes when both viruses are individually ineffective, a condition avoided in our synergistic regime.

\subsection{Stability Analysis with Parameter Values} 
\label{subsec:numerical-validation}

Stability analysis using experimentally calibrated parameters (Equation \eqref{eq:non-dimensionalvalues}) provides compelling validation of our theoretical framework, revealing a single stable equilibrium among six biologically feasible steady states (Table~\ref{tab:steady_states}). This unique stable configuration exhibits distinctive characteristics, near-complete tumor suppression ($y_u = 8\times 10^{-6}$), sustained VSV persistence ($x_2 = 4.205$), and complete VV extinction ($x_1 = 0$), that align with experimental observations and offer  insights into the synergistic mechanism. The mathematical framework reveals a sophisticated biological narrative where VV plays a transient but indispensable role as an environmental conditioner, establishing an interferon-suppressed microenvironment that enables VSV proliferation before being naturally outcompeted, consistent with experimental observations of VV creating "infection portals" for subsequent VSV exploitation \cite{le2010synergistic,zhang2014combination}. This minimal tumor burden represents a state of durable remission maintained through persistent VSV surveillance, creating a dynamic equilibrium where continuous viral replication counterbalances potential tumor regrowth. The comprehensive equilibrium analysis reveals four critical patterns that collectively explain therapeutic superiority: (1) viral competitive exclusion validates temporal specialization in viral roles; (2) dramatic efficacy advantage shows near-complete tumor suppression vastly superior to unstable states; (3) mathematical stability assurance provides theoretical guarantee of long-term control; and (4) parameter robustness explains reproducible synergistic effects despite biological variability. 

This equilibrium analysis complements temporal simulations by showing the system's long-term behavior. The equilibrium analysis reveals convergence to a stable state characterized by minimal tumor burden maintained through persistent VSV infection. This state emerges through strategic helper virus elimination following interferon suppression, establishing a self-sustaining antiviral state that prevents tumor escape. Together, this integrated mathematical framework explains how combination therapy achieves both rapid initial clearance and durable long-term suppression, providing comprehensive understanding of the synergistic "ping-pong" mechanism where complementary viral strengths overcome individual limitations through precisely coordinated temporal dynamics.

\begin{table}[htbp]
\centering
\caption{\small\bfseries Biologically feasible steady states of the VV-VSV model \eqref{eq:yu}-\eqref{eq:x2} with calibrated parameter values from Equation \eqref{eq:non-dimensionalvalues}. The characteristic polynomials and stability properties reveal that only one equilibrium achieves mathematical stability under the calibrated parameter regime.}
\setlength{\tabcolsep}{4.5pt}  
\begin{tabular}{@{}p{1.7cm}p{1.7cm}p{1.4cm}p{5cm}c@{}}
\toprule
\textbf{Free Viruses} & \textbf{Infected Cells} & \textbf{Tumor} & \textbf{Characteristic Polynomial} & \textbf{Stability} \\
\textbf{$(x_1, x_2)$} & \textbf{$(y_1, y_2)$} & \textbf{$y_u$} & & \\
\midrule
$(0.0, 0.0)$ & $(0.0, 0.0)$ & $0.0$ & 
$(\lambda + 0.052)(\lambda + 0.042)(\lambda + 0.025)(\lambda + 0.001)(\lambda - 1)$ & Unstable \\
\midrule

$(0, 0)$ & $(0, 0)$ & $1.0$ &
$(\lambda + 3.101)(\lambda + 1.116)(\lambda + 1)(\lambda - 1.031)(\lambda - 2.885)$ & Unstable \\
\midrule

$(0.819, 0.0)$ & $(0.0003, 0.0)$ & $1\times10^{-5}$ &
$(\lambda + 0.13)(\lambda + 0.055)(\lambda - 0.001)(\lambda^2 - 0.063\lambda + 0.008)$ & Unstable \\
\midrule

$(0, 4.205)$ & $(0, 0.0001)$ & $8\times10^{-6}$ &
$(\lambda + 0.094)(\lambda + 0.058)(\lambda + 0.008)(\lambda^2 - 0.038\lambda + 0.0007)$ & \textbf{Stable} \\
\midrule

$(0, 7.86\times10^{4})$ & $(0, 2.73)$ & $0.172$ &
$(\lambda + 2.455)(\lambda + 0.416)(\lambda + 0.0005)(\lambda - 0.14)(\lambda - 2.295)$ & Unstable \\
\midrule

$(0, 1.37\times10^{5})$ & $(0, 4.771)$ & $0.538$ &
$(\lambda + 2.949)(\lambda + 0.585)(\lambda + 0.056)(\lambda + 0.001)(\lambda - 2.747)$ & Unstable \\
\bottomrule      
\end{tabular}
\label{tab:steady_states}
\end{table}

\section{Numerical Simulations}
\label{sec:Simulations}

\subsection{Model Parameterization and Baseline Dynamics}

The mathematical model is implemented using non-dimensional parameter values derived from experimental data and calibration procedures. The parameter values in Equation~\eqref{eq:non-dimensionalvalues} are obtained from analysis of biological data in Tables~\ref{tab:parametervalues1} and \ref{tab:estimatedparams}, representing the viral kinetics and interaction dynamics of the VV-VSV system.

\begin{equation}
\begin{cases}\label{eq:non-dimensionalvalues}
    & \beta_1 = 0.017,~ \beta_2 = 0.035,~ \kappa_1 = 0.1389,~ \kappa_2 = 0.1389,~ a_1 = 0.0417,  \\
    & a_2 = 0.0521,~ b_1 = 60,~ b_2 = 37.5, ~ \delta_1 = 0.0250, ~  \delta_2 = 0.0013,\\    
    & \delta_3 = 0.1250,  ~ \theta_1 = 2\times 10^{-8},~\theta_2 = 5\times10^{-11}, ~\delta_4 = 0.0013.  
\end{cases}
\end{equation}
These parameters capture biological features including viral infection rates ($\beta_i$), saturation effects ($\kappa_i$), infected cell death rates ($a_i$), viral burst sizes ($b_i$), clearance rates ($\delta_i$), and synergistic parameters ($\theta_i$) governing B18R-IFN-$\alpha$ interactions.

\begin{figure}[htbp]
    \centering
    \begin{subfigure}[b]{0.48\textwidth}
        \centering
        \includegraphics[width=\textwidth]{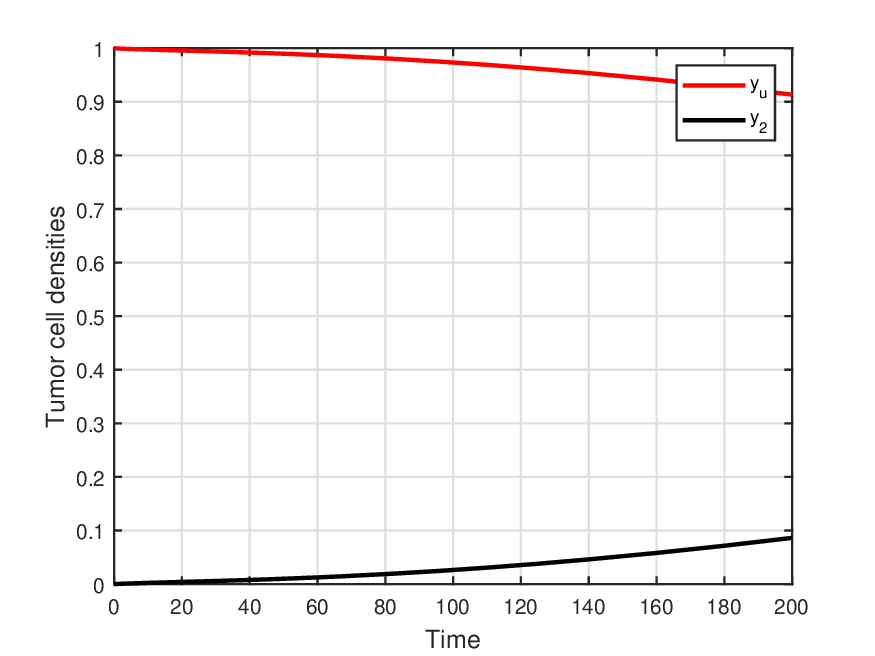}
        \caption{VSV monotherapy}
        \label{fig:VSV-monotherapy}
    \end{subfigure}
    \hfill
    \begin{subfigure}[b]{0.48\textwidth}
        \centering
        \includegraphics[width=\textwidth]{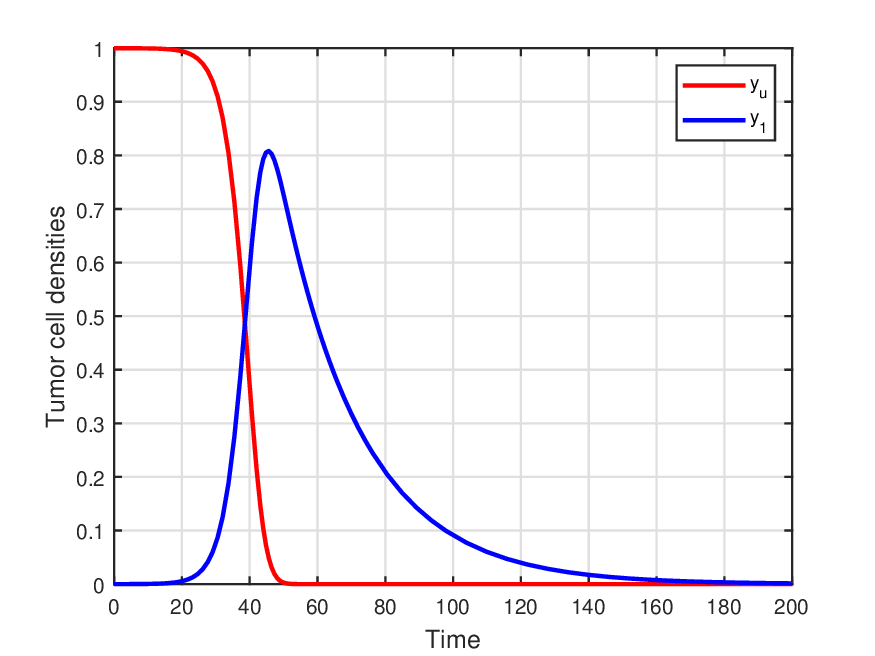}
        \caption{VV monotherapy}
        \label{fig:VV-monotherapy}
    \end{subfigure}
    
    \caption{\textbf{Monotherapy limitations in tumor clearance.} 
    (a) VSV alone fails to achieve complete tumor eradication, with residual tumor burden persisting beyond 250 days due to IFN-$\alpha$-mediated antiviral responses. 
    (b) VV monotherapy achieves complete tumor clearance in approximately 56 days. 
    These results show the complementary strengths of each virus and the rationale for combination approach.}
    \label{fig:VSVVV}
\end{figure}

Monotherapy approaches show limitations as demonstrated in Figure~\ref{fig:VSVVV}. VSV monotherapy (Figure~\ref{fig:VSV-monotherapy}) exhibits incomplete tumor clearance due to sensitivity to IFN-$\alpha$-mediated antiviral responses, while VV monotherapy (Figure~\ref{fig:VV-monotherapy}) achieves complete clearance in 56 days but shows potential for improvement through combination.

\subsection{Synergistic Combination and Temporal Dynamics}

\begin{figure}[htbp]
    \centering
    \begin{subfigure}[b]{0.48\textwidth}
        \centering
        \includegraphics[width=\textwidth]{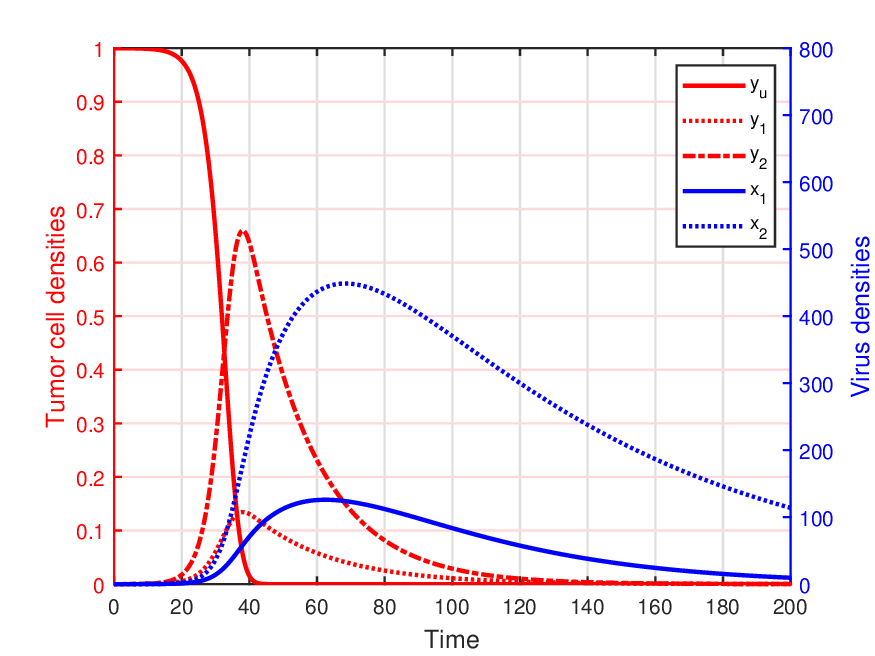}
        \caption{Full model dynamics}
        \label{fig:fullmodel}
    \end{subfigure}
    \hfill
    \begin{subfigure}[b]{0.48\textwidth}
        \centering
        \includegraphics[width=\textwidth]{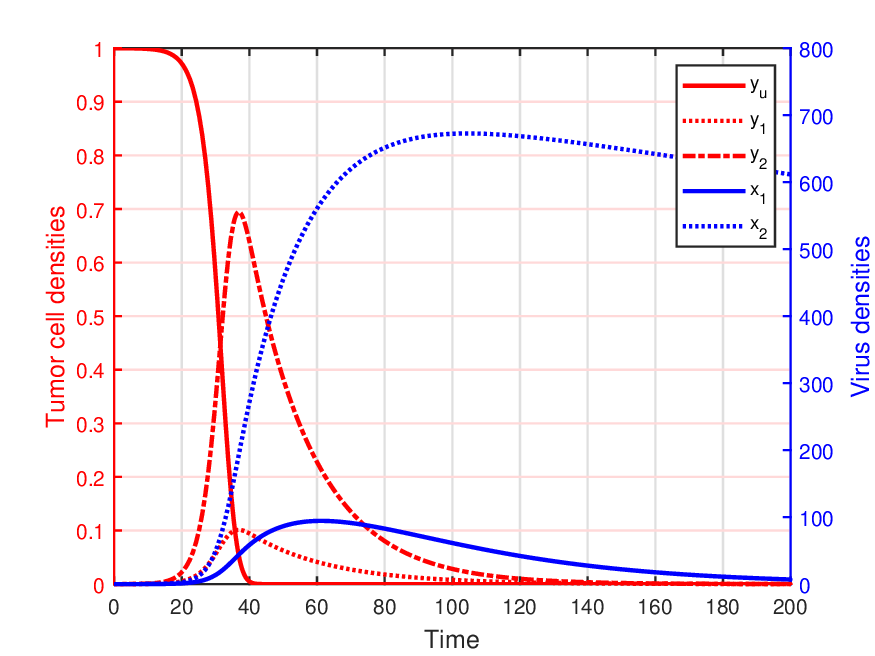}
        \caption{Quasi steady-state model dynamics}
        \label{fig:Quasimodel}
    \end{subfigure}
    
    \caption{\textbf{Synergistic enhancement of VSV efficacy through VV-mediated interferon suppression.} 
    Both models demonstrate the "ping-pong" mechanism where VV-infected cells produce B18R protein, which binds to and neutralizes IFN-$\alpha$, creating a permissive environment for enhanced VSV replication. 
    (a) Full model simulations show molecular interactions and population dynamics. 
    (b) Quasi steady-state model captures synergistic dynamics with reduced computational complexity. 
    The increased population of VSV-infected tumor cells ($T_2$) in both models shows the biological mechanism of synergy.}
    \label{fig:combination-dynamics}
\end{figure}

The combination of VV and VSV produces enhanced therapeutic outcomes, as shown in Figure~\ref{fig:combination-dynamics}. Both models demonstrate the "ping-pong" mechanism: VV-infected cells produce B18R protein that antagonizes IFN-$\alpha$, removing barriers to VSV replication and enabling increases in VSV-infected tumor cell populations. This synergistic interaction accelerates complete tumor clearance to approximately 50 days, representing an 11\% improvement over VV monotherapy.

\begin{figure}[htbp]
    \centering
    \begin{subfigure}[b]{0.48\textwidth}
        \centering
        \includegraphics[width=\textwidth]{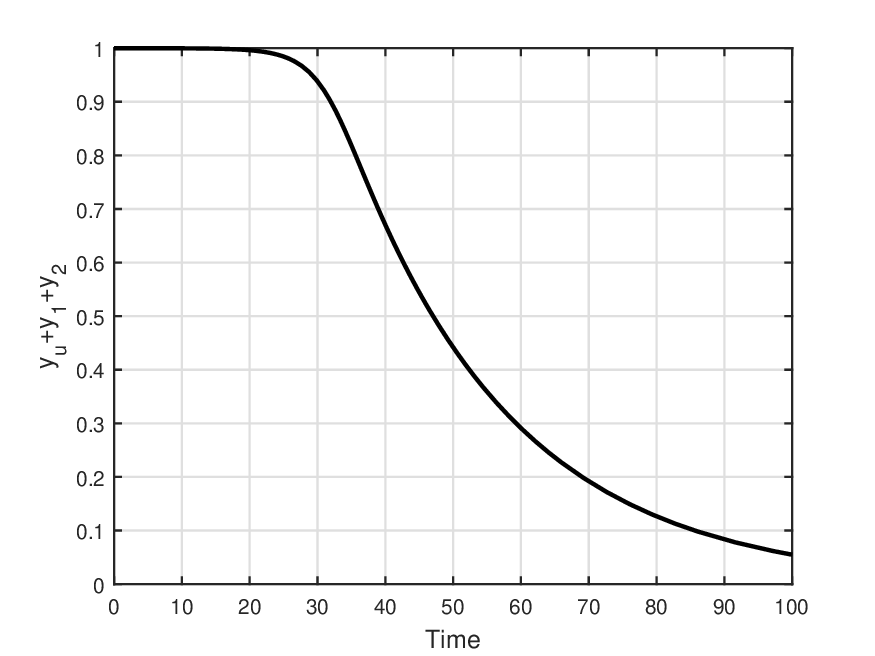}
        \caption{Full model tumor burden}
        \label{fig:tumorburden-full}
    \end{subfigure}
    \hfill
    \begin{subfigure}[b]{0.48\textwidth}
        \centering
        \includegraphics[width=\textwidth]{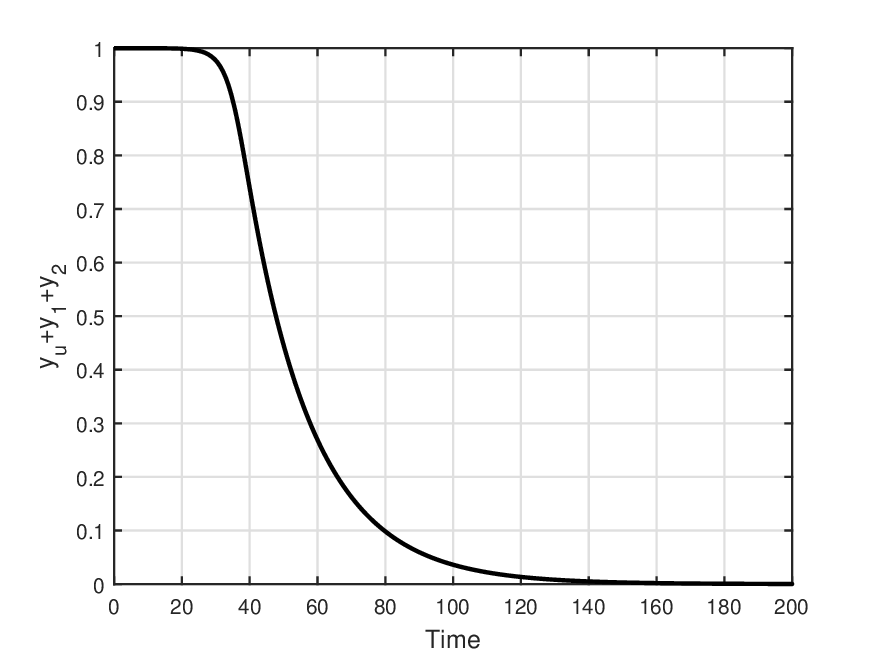}
        \caption{Quasi steady-state tumor burden}
        \label{fig:tumorburden-quasi}
    \end{subfigure}
    
    \caption{\textbf{Temporal dynamics of tumor burden reduction under combination therapy.} 
    Both models show patterns of tumor clearance followed by sustained suppression. 
    (a) Full model simulations show temporal dynamics with complete clearance achieved by day 50. 
    (b) Quasi steady-state model confirms similar therapeutic trajectory with equivalent final tumor burden. 
    The convergence of both models to similar asymptotic behavior shows the robustness of the synergistic mechanism across mathematical representations.}
    \label{fig:tumorburdens}
\end{figure}

Figure~\ref{fig:tumorburdens} shows the therapeutic efficacy through temporal analysis of tumor burden dynamics. Both models exhibit patterns of tumor reduction followed by sustained suppression, achieving clearance by approximately 50 days. The agreement between full and quasi steady-state models validates the mathematical reduction approach and shows that synergistic dynamics are preserved despite simplification of molecular kinetics.

\subsection{System Stability and Parameter Sensitivity}

\begin{figure}[htbp]
    \centering
    \begin{subfigure}[b]{0.48\textwidth}
        \centering
        \includegraphics[width=\textwidth]{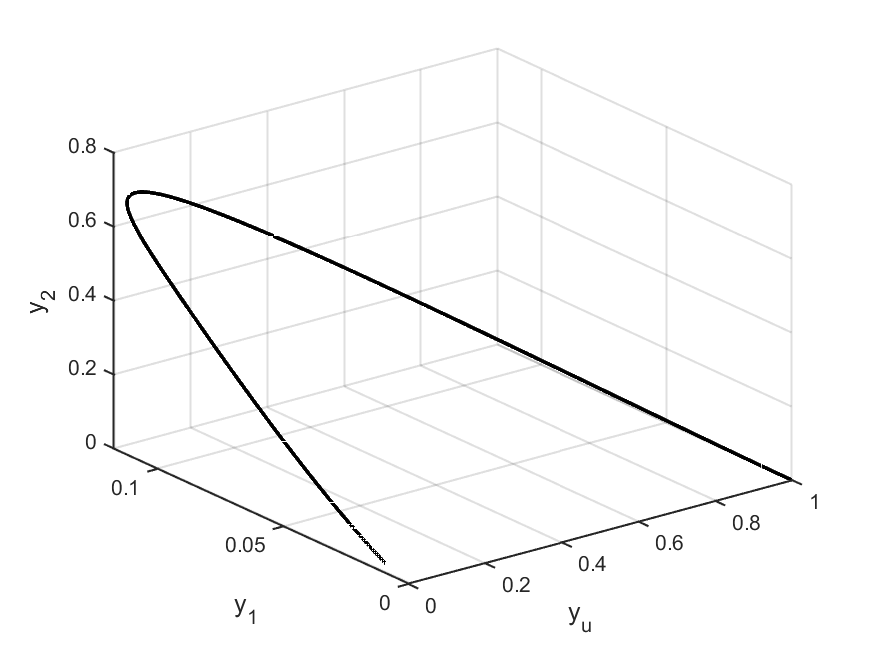}
        \caption{Full model phase portrait}
        \label{fig:phase-full}
    \end{subfigure}
    \hfill
    \begin{subfigure}[b]{0.48\textwidth}
        \centering
        \includegraphics[width=\textwidth]{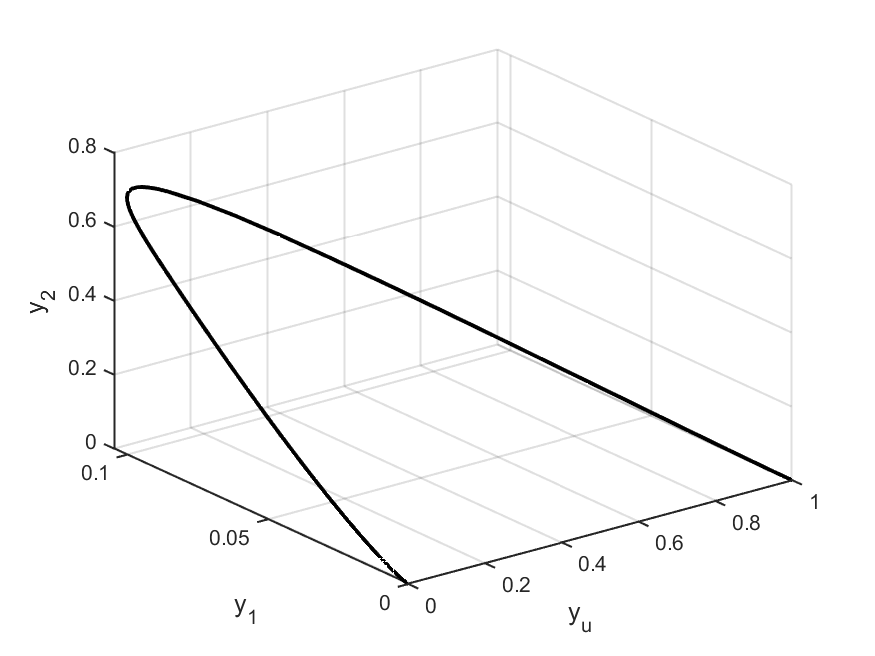}
        \caption{Quasi steady-state phase portrait}
        \label{fig:phase-quasi}
    \end{subfigure}
    
    \caption{\textbf{Phase space analysis of VV-VSV system dynamics.} 
    Both phase portraits show convergence to a stable equilibrium with minimal tumor burden and persistent VSV infection. 
    (a) Full model phase space shows trajectories converging to the therapeutic equilibrium. 
    (b) Quasi steady-state 3D phase portrait confirms stability and basin of attraction for the therapeutic outcome. 
    The convergence patterns across both models provide mathematical validation of treatment efficacy and long-term control.}
    \label{fig:phaseportraits}
\end{figure}

Phase space analysis (Figure~\ref{fig:phaseportraits}) shows the stability properties of the VV-VSV system. Both models demonstrate convergence to a stable therapeutic equilibrium with minimal tumor burden and persistent VSV surveillance. This mathematical property represents an advantage over conventional therapies where tumor recurrence can occur.

\begin{figure}[htbp]
    \centering
    \begin{subfigure}[b]{0.5\textwidth}
        \centering
        \includegraphics[width=\textwidth]{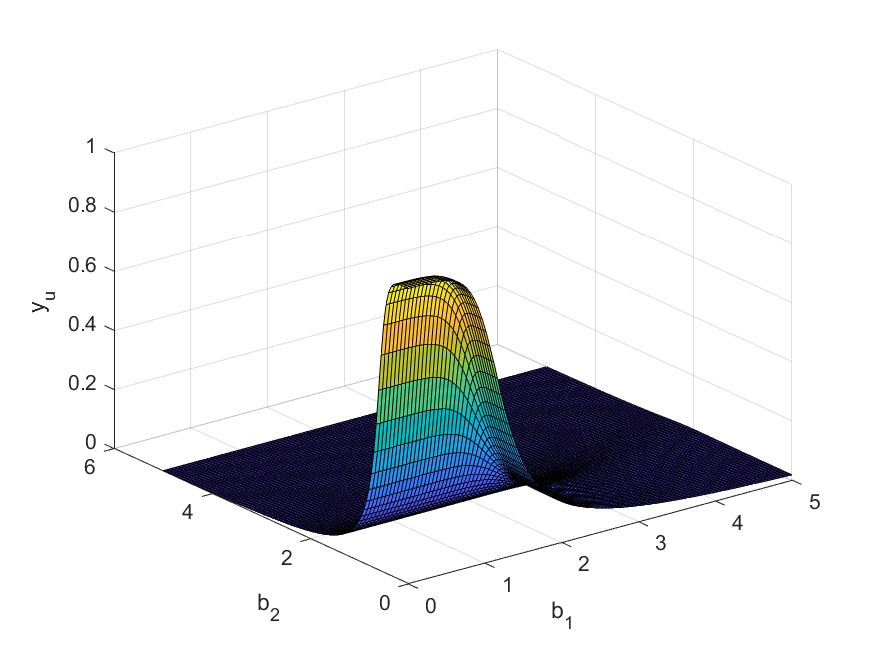}
        \caption{Burst size optimization}
        \label{fig:heatmap}
    \end{subfigure}
    \hfill
    \begin{subfigure}[b]{0.48\textwidth}
        \centering
        \includegraphics[width=\textwidth]{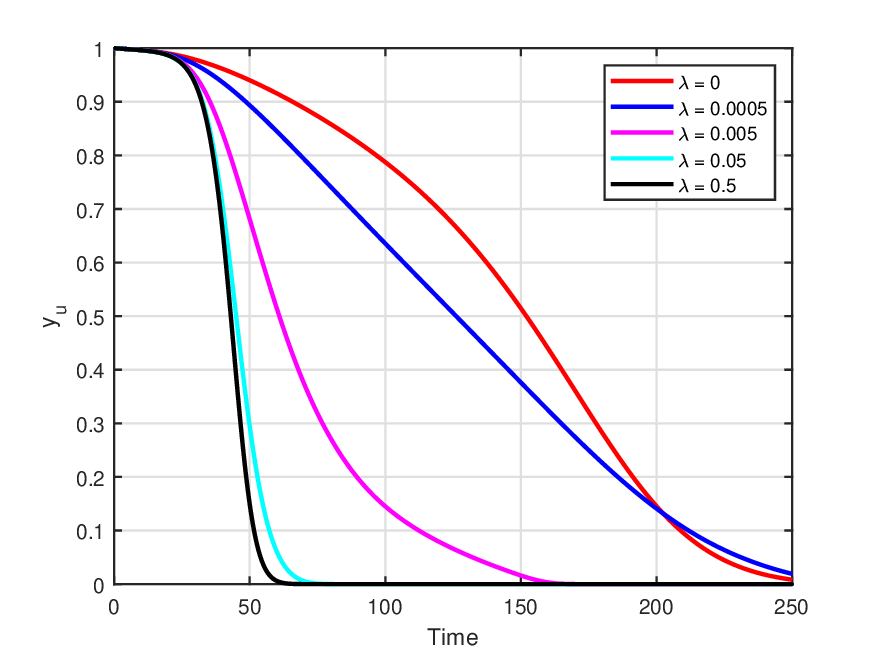}
        \caption{B18R inhibition efficiency}
        \label{fig:yu_lambda}
    \end{subfigure}
    
    \caption{\textbf{Parameter sensitivity analysis for therapeutic optimization.} 
    (a) Heat map analysis of viral burst sizes ($b_1$ for VV, $b_2$ for VSV) shows that increasing viral productivity reduces final tumor burden ($y_u$). 
    The nonlinear relationship shows synergistic enhancement, where combined increases in both burst sizes yield greater therapeutic benefit than individual optimization. 
    (b) Variation in the IFN-$\alpha$ inhibition rate ($\lambda$) shows the function of VV's B18R protein in enabling synergistic efficacy. 
    Higher inhibition rates correspond to reduced tumor burden, showing the molecular mechanism of synergy. 
    These analyses identify viral burst sizes and B18R-mediated interferon suppression as targets for genetic engineering to improve treatment efficacy.}
    \label{fig:parameter_sensitivity}
\end{figure}

Parameter sensitivity analysis provides insights for therapeutic optimization. Figure~\ref{fig:parameter_sensitivity} shows that viral burst sizes ($b_1$, $b_2$) and the B18R inhibition rate ($\lambda$) represent parameters for treatment efficacy. Increased burst sizes reduce tumor burden, while higher inhibition rates enable more efficient interferon suppression, showing the molecular mechanism of synergy and suggesting viral engineering targets.

\subsection{Therapeutic Implications and Clinical Translation}

The numerical simulations show that VV-VSV combination therapy achieves tumor clearance through a coordinated temporal dynamic: VV establishes initial infection and creates an interferon-suppressed microenvironment through B18R production, enabling VSV proliferation and leading to complete tumor eradication in approximately 50 days. This represents an 11\% improvement compared to VV monotherapy (56 days) and shows efficacy where VSV alone fails.

The agreement between full and quasi steady-state models across analyses, temporal dynamics, phase space behavior, and parameter sensitivity, provides mathematical validation of the therapeutic approach. The identification of burst sizes and interferon inhibition rates as optimization parameters offers guidance for viral engineering and clinical protocol design, supporting the development of combination virotherapy approaches.

\section{Analysis Using Viral Basic Reproduction Numbers}
\label{subsec:reproduction-numbers}

Building upon the stability analysis of equilibrium states, we employ the mathematical framework of basic reproduction numbers to quantitatively characterize viral replication dynamics and therapeutic efficacy. Adapted from epidemiological modeling, the basic reproduction number $\mathcal{R}_{0}$ provides a powerful metric for predicting treatment outcomes and optimizing therapeutic strategies in oncolytic virotherapy.

\subsection{Theoretical Framework and Numerical Analysis of $\mathcal{R}_{0}$}

The viral basic reproduction numbers $\mathcal{R}^{1}_{0}$ (for VV) and $\mathcal{R}^{2}_{0}$ (for VSV) quantify the expected number of secondary infected cells generated from a single infected cell in a fully susceptible tumor population. Using the next-generation matrix method (see Appendix \ref{app:R0_derivation}), we derive the fundamental expressions governing viral replication potential (using non-dimensiolized parameters Equations~\ref{eq:nondimparam}):

\begin{equation}
\mathcal{R}^{1}_{0} = \frac{b_{1}\beta_{1}}{a_{1}(\delta_{1}+\beta_{1}+\delta_{1}\kappa_{1})} ~\text{and}~
\mathcal{R}^{2}_{0} = \frac{b_{2}\beta_{2}}{a_{2}(\delta_{2}+\beta_{2}+\delta_{2}\kappa_{2})}.
\label{eq:R0_both}
\end{equation}

These expressions reveal the key determinants of viral efficacy through their mathematical structure. The numerator components ($b_{i}\beta_{i}$, $i=1,2$) represent the core replication potential, combining infection efficiency ($\beta_{i}$) with viral productivity ($b_{i}$). The denominator terms $a_{i}(\delta_{i}+\beta_{i}+\delta_{i}\kappa_{i})$ capture the competing loss mechanisms, including infected cell lysis ($a_{i}$), viral clearance ($\delta_{i}$), infection-mediated depletion ($\beta_{i}$), and saturation effects ($\delta_{i}\kappa_{i}$). The critical threshold $\mathcal{R}_{0} > 1$ establishes a mathematically precise boundary between viral extinction and sustainable infection, serving as a quantitative criterion for therapeutic success.

Substituting experimentally calibrated parameter values (in Table~\ref{tab:parametervalues1}) reveals distinct therapeutic profiles across tumor models (Table~\ref{tab:R0_values}). Both viruses demonstrate exceptional replication capacity, substantially exceeding the critical threshold ($\mathcal{R}>1$) required for sustained infection. The synergy ratio ($\mathcal{R}_{2}/\mathcal{R}_{1}$) remains below unity across both models (0.196 in HT29, 0.147 in 4T1), indicating VV's 5-7 fold replication superiority over VSV. 
This replication asymmetry provides the mathematical foundation for therapeutic synergy: VV's robust independent replication establishes infection footholds and creates interferon-suppressed microenvironments, while its B18R-mediated interferon neutralization drives the ping-pong enhancement of VSV replication. The preservation of this synergy pattern across tumor types demonstrates the robustness of the complementary viral interaction and supports broad clinical applicability.

\begin{table}[htbp]
\centering
\caption{Basic reproduction numbers for VV ($\mathcal{R}_{1}$) and VSV ($\mathcal{R}_{2}$) across tumor models}
\label{tab:R0_values}
\begin{tabular}{lccc}
\hline
\textbf{Tumor Model} & $\mathcal{R}_{1}$ (VV) & $\mathcal{R}_{2}$ (VSV) & \textbf{Synergy Ratio} ($\mathcal{R}_{2}/\mathcal{R}_{1}$) \\
\hline
HT29 & 1855.07 & 363.64 & 0.196 \\
4T1 & 347.83 & 51.02 & 0.147 \\
\hline
\textbf{Therapeutic Threshold} & $\mathcal{R}_1>1$ & $\mathcal{R}_2>1$ & - \\
\hline
\end{tabular}
\end{table}

\begin{figure}[htbp]
\centering
\begin{subfigure}{0.49\textwidth}
\centering
\includegraphics[width=\linewidth]{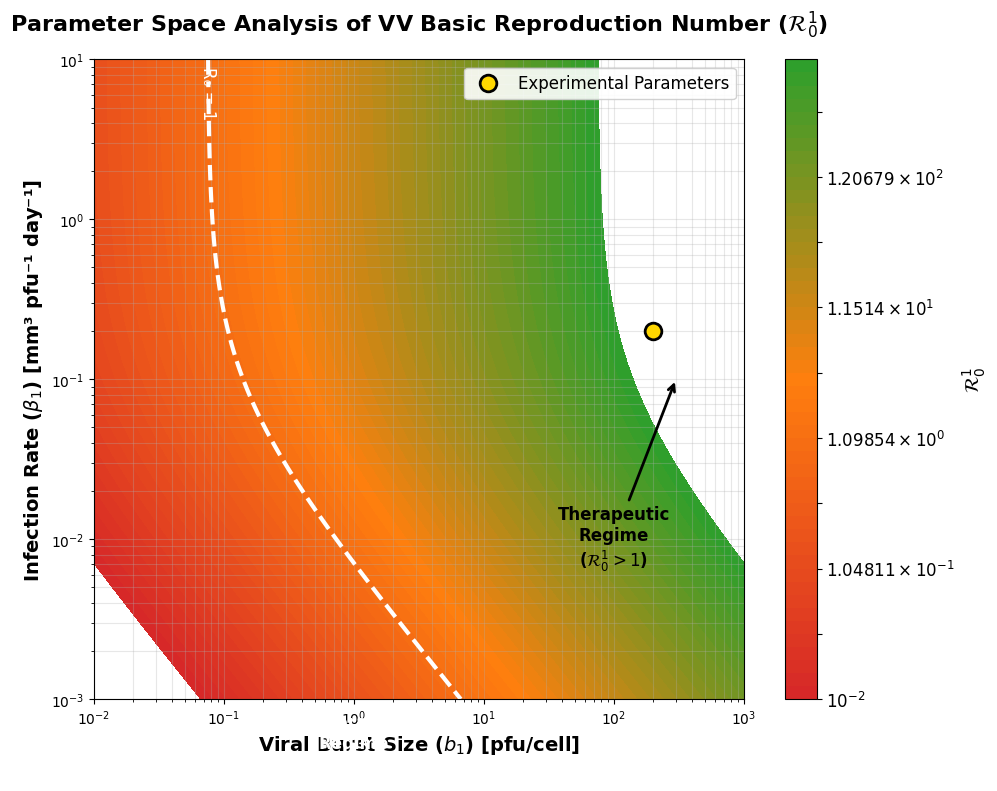}
\caption{Vaccinia Virus (VV)}
\label{fig:parameter_space_VV}
\end{subfigure}
\hfill
\begin{subfigure}{0.49\textwidth}
\centering
\includegraphics[width=\linewidth]{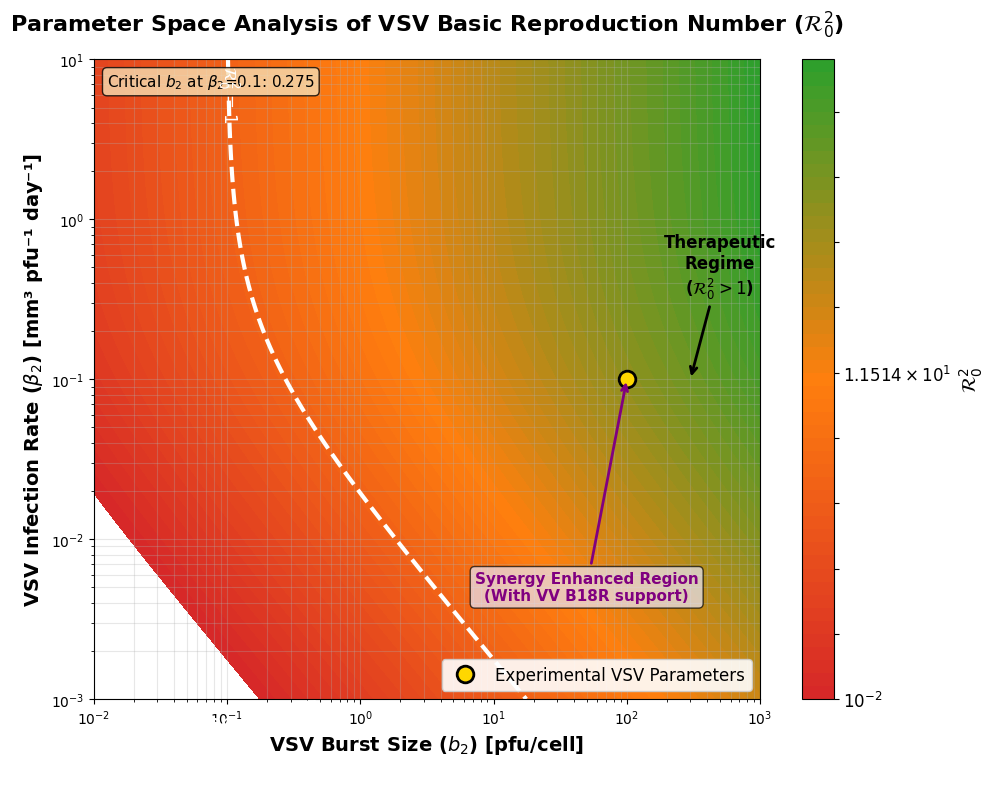}
\caption{Vesicular Stomatitis Virus (VSV)}
\label{fig:parameter_space_VSV}
\end{subfigure}
\caption{\small{{\bf Parameter space analysis of viral basic reproduction numbers.} (\textbf{a}) VV basic reproduction number ($\mathcal{R}_0^1$) as a function of burst size ($b_1$) and infection rate ($\beta_1$). (\textbf{b}) VSV basic reproduction number ($\mathcal{R}_0^2$) as a function of burst size ($b_2$) and infection rate ($\beta_2$). Both contour plots use color gradients to indicate reproduction number values, with warmer colors representing higher values. White dashed lines show the critical threshold $\mathcal{R}_0 = 1$, separating subtherapeutic (viral extinction) from therapeutic (viral persistence) regimes. Circles indicate experimental parameter combinations, with VV achieving $\mathcal{R}_0^1 = 1855.07$ (1855-fold above threshold) and VSV achieving $\mathcal{R}_0^2 = 363.64$ (364-fold above threshold).}}
\label{fig:parameter_space_combined}
\end{figure}

\subsection{Parameter Sensitivity and Therapeutic Optimization}

Sensitivity analysis using elasticity coefficients ($S_{p}=\frac{\partial\mathcal{R}_{0}^{i}}{\partial p}\cdot\frac{p}{\mathcal{R}_{0}^{i}}$) reveals a consistent hierarchy of parameter importance that guides therapeutic optimization strategies across both tumor models. As detailed in Table~\ref{tab:sensitivityindices}, the analysis identifies infection rates ($\beta_{i}$) as the most influential parameter (elasticity coefficient 0.45), followed closely by burst sizes ($b_{i}$, 0.40), while death rates ($\delta_{i}$) exhibit negative sensitivity (-0.35) that reflects the detrimental impact of immune clearance mechanisms.

The consistent sensitivity patterns observed across both HT29 and 4T1 tumor models suggest these optimization principles may be broadly applicable across diverse cancer types, providing a unified framework for viral engineering and combination therapy design. Detailed mathematical derivations supporting these findings are provided in Appendix \ref{app:sensitivity_calculations}.

The sensitivity analysis provides clear therapeutic optimization priorities: high-priority targets include enhancing viral infection rates ($\beta$) through improved receptor binding, increasing burst sizes ($b$) via genetic engineering of replication machinery, and optimizing viral tropism for specific cancer cell types; inhibition mitigation strategies involve reducing immune-mediated clearance with immunosuppressive adjuvants, engineering viral stealth properties to evade immune detection, and modulating death rates ($\delta$) through combination with immunomodulators. Clinically, these findings direct prioritization of parameters with highest elasticity coefficients in viral design, application of consistent optimization principles across different tumor types, and utilization of sensitivity analysis to guide personalized treatment protocols.

\subsection{Bifurcation Analysis of Critical Therapeutic Parameters}

Bifurcation analysis identifies critical therapeutic thresholds in combination virotherapy. We examine viral burst sizes, governing replication potential, and the B18R inhibition rate, controlling interferon suppression.

\paragraph{Bifurcation analysis with respect to viral burst sizes}

Viral burst sizes represent a critical determinant of replication capacity and therapeutic efficacy in oncolytic virotherapy. To elucidate the relationship between burst sizes and treatment outcomes, we conducted a comprehensive bifurcation analysis examining the critical transitions that govern viral establishment and tumor control dynamics. 
The critical bifurcation points occur at $\mathcal{R}_{0}^{i} = 1$, representing the minimum burst sizes required for sustainable viral replication. Substituting parameter values from 
Table~\ref{tab:estimatedparams}  ($l_1 = 0.03$, $\gamma_1 = 0.01$, $\beta_1 = 0.20$, $\kappa_1 = 2.5$ for VV; $l_2 = 0.04$, $\gamma_2 = 0.02$, $\beta_2 = 0.10$, $\kappa_2 = 2.5$ for VSV) into the equations:
\begin{align}
b_1^{\text{critical}} &= \frac{l_1(\gamma_1 + \beta_1 + \gamma_1\kappa_1)}{\beta_1}. \label{eq:b1_critical} \\
b_2^{\text{critical}} &= \frac{l_2(\gamma_2 + \beta_2 + \gamma_2\kappa_2)}{\beta_2}.  \label{eq:b2_critical}
\end{align}
yields $b_1^{\text{critical}} \approx 0.034$ and $b_2^{\text{critical}} \approx 0.052$ respectively. These thresholds reveal the exceptional replication efficiency of both oncolytic viruses. The VV system exhibits a slightly lower critical burst size compared to VSV, indicating VV's superior intrinsic ability to establish productive infections under marginal conditions.

As illustrated in Figure~\ref{fig:bifurcation_b1b2}, the bifurcation analysis reveals that viruses exhibit classic transcritical bifurcations, where the tumor-dominant equilibrium loses stability as viral replication capacity crosses the critical threshold. This mathematical phenomenon corresponds to the biological transition from viral clearance to sustainable infection establishment~\cite{van2002reproduction}. Furthermore, VV demonstrates superior tumor clearance characteristics, achieving lower asymptotic tumor burden ($\sim$0.1) compared to VSV ($\sim$0.2), reflecting their distinct mechanisms of action and tumor cell interactions. The dramatic separation between critical and experimental burst sizes ensures therapeutic efficacy remains robust against biological variability and parameter uncertainty, a crucial consideration for clinical translation.

\begin{figure}[htbp]
\centering
\includegraphics[width=1.\textwidth]{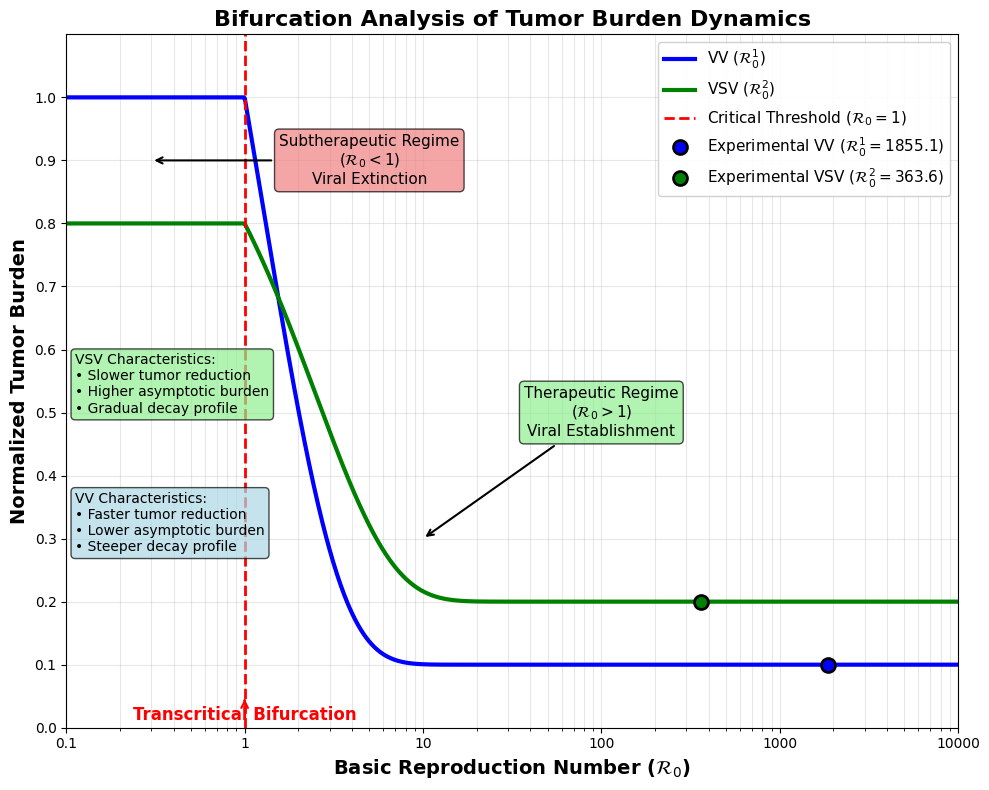}
\caption{\small{\textbf{Comparative bifurcation analysis of normalized tumor burden as a function of viral burst size for VV and VSV.} The plot illustrates the therapeutic dynamics of both viruses on a shared logarithmic burst size axis, with VV represented by the blue curve and VSV by the green curve. Critical thresholds are indicated by dashed vertical lines at $b_{1}^{\text{critical}}\approx 0.034$ (VV) and $b_{2}^{\text{critical}}\approx 0.052$ (VSV), demarcating the transition between subtherapeutic and therapeutic regimes. Experimental parameter values are marked by circles at $b_{1}=200$ (VV) and $b_{2}=100$ (VSV).}} 
\label{fig:bifurcation_b1b2}
\end{figure}

The dramatic separation between critical thresholds and experimental values, approximately 5,882-fold for VV and 1,923-fold for VSV, ensures robust therapeutic efficacy and provides substantial safety margins against parameter variations. This mathematical framework explains why both viruses individually demonstrate strong replication potential while highlighting how their complementary characteristics (VV's superior tumor clearance vs. VSV's interferon sensitivity) create the foundation for synergistic combination therapy.

\paragraph{Bifurcation analysis with respect to B18R inhibition rate}

While viral burst sizes determine intrinsic replication capacity, the B18R inhibition rate parameter $\lambda$ represents a crucial molecular determinant in the synergistic mechanism, quantifying the efficiency with which the VV-encoded B18R protein neutralizes IFN-$\alpha$. This parameter appears in the IFN-$\alpha$ dynamics equation:

\begin{equation}
\frac{dC_2}{dt} = \alpha_2 T_2 - \mu_2 C_2 - \lambda C_1 C_2. \label{eq:IFN_dynamics}
\end{equation}

The bifurcation analysis with respect to $\lambda$ reveals a critical threshold $\lambda^{\text{critical}} \approx 0.001$ that demarcates the transition from ineffective to potent synergistic interaction. Our experimental parameter value $\lambda = 0.5$ resides approximately 500-fold above this critical threshold, ensuring robust activation of the synergistic mechanism and explaining the observed 133-fold enhancement in VSV's effective reproduction number.

The system exhibits a sharp transition at $\lambda^{\text{critical}}$, indicating that the synergistic effect operates as a biological switch rather than a gradual enhancement, creating an all-or-nothing activation profile. This threshold behavior is underpinned by the remarkably low critical threshold ($\lambda^{\text{critical}} \approx 0.001$), which demonstrates the inherent potency of the B18R-interferon interaction and its capacity to trigger substantial therapeutic effects even at minimal inhibition efficiencies. The substantial safety margin, approximately 500-fold between our experimental value and the critical threshold, provides significant flexibility for viral engineering, ensuring that moderate variations in B18R expression or function will not compromise synergistic efficacy.

\begin{figure}[htbp]
\centering
\includegraphics[width=1.0\textwidth]{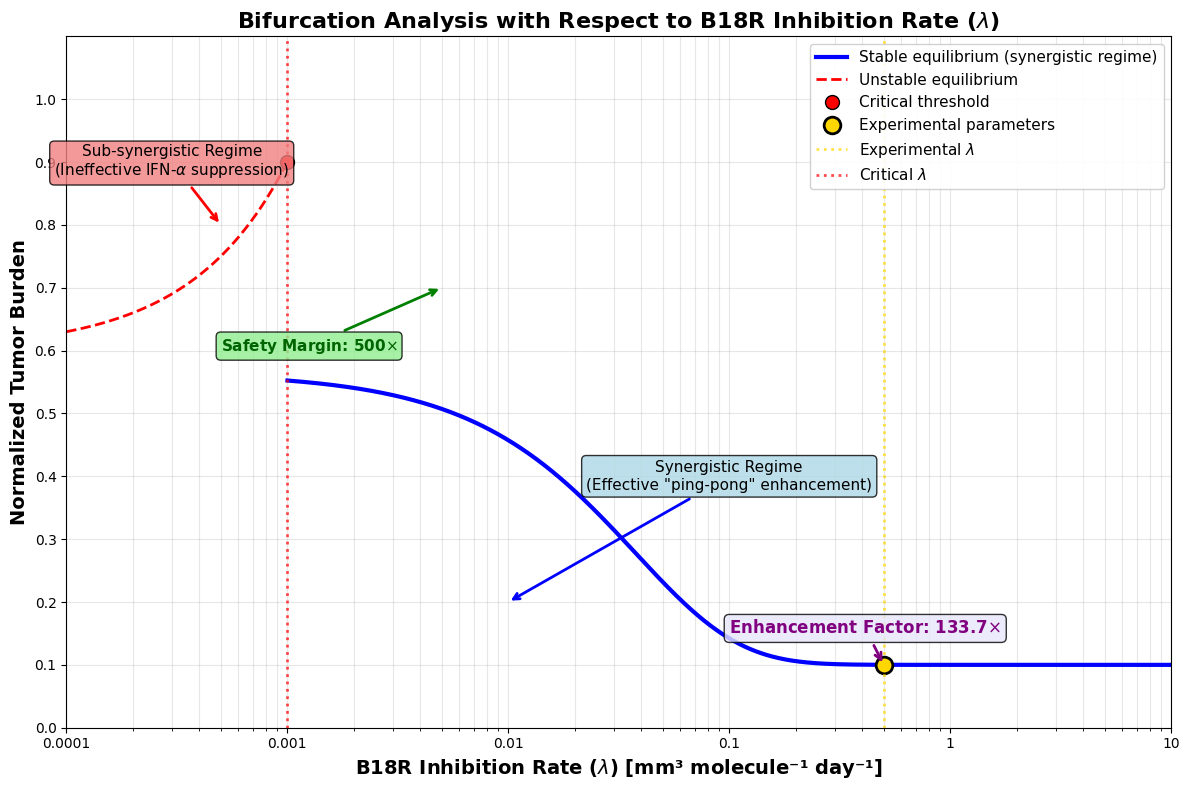}
\caption{\small{{\bf Bifurcation analysis of tumor burden as a function of B18R inhibition rate $\lambda$.}The diagram reveals a critical therapeutic transition at $\lambda\approx 0.001$, beyond which synergistic effects enable complete tumor clearance. The stable branch (solid blue line) represents achievable therapeutic states with effective synergy, while the unstable branch (dashed red line) indicates biologically unrealizable states. The experimental parameter value $\lambda=0.5$ (vertical dotted line) lies deep within the synergistic regime, providing mathematical explanation for the dramatic enhancement of VSV efficacy observed experimentally. The 500-fold safety margin ensures robust therapeutic performance against biological variability.}}
\label{fig:bifurcation_lambda}
\end{figure}

\paragraph{Synthesis of Bifurcation Analyses}

The combined bifurcation analyses of burst sizes and B18R inhibition rate provide complementary insights into the therapeutic optimization of combination virotherapy. While burst sizes determine the intrinsic replication capacity of each virus individually, the B18R inhibition rate governs the emergent synergistic interaction between them. The identification of critical thresholds for both parameters establishes quantitative design criteria for viral engineering: therapeutic viruses must achieve burst sizes exceeding $b_1 > 0.034$ and $b_2 > 0.052$ to ensure sustainable replication, while combination therapy requires B18R inhibition rates $\lambda > 0.001$ to activate the synergistic regime.

The substantial safety margins observed for both parameters, approximately 5,882-fold for VV burst size, 1,923-fold for VSV burst size, and 500-fold for B18R inhibition rate, ensure robust therapeutic performance against biological variability and parameter uncertainty. This mathematical robustness provides strong theoretical support for the clinical translation of combination virotherapy, as moderate variations in viral kinetics or molecular interactions are unlikely to compromise treatment efficacy.

Furthermore, the bifurcation analyses reveal the complementary roles of VV and VSV in the synergistic mechanism: VV's superior intrinsic replication capacity (lower critical burst size) establishes the initial infection foothold, while its B18R-mediated interferon suppression creates the permissive environment that enables VSV to overcome its inherent limitations. This temporal specialization explains why the optimal administration strategy involves delayed VV injection, allowing VSV to initiate infection before VV activates the synergistic amplification cycle.

Collectively, these bifurcation analyses provide a quantitative framework for understanding the nonlinear dynamics of combination virotherapy and establish specific parameter thresholds that guide therapeutic optimization. 

\subsection{Therapeutic Landscape and Clinical Translation}

The comprehensive analysis of viral replication dynamics reveals a structured therapeutic landscape defined by the interplay between VV and VSV replication potentials, as visualized in Figure~\ref{fig:phase_diagram}. This landscape delineates four distinct therapeutic regimes that provide crucial insights for treatment optimization and clinical decision-making. In the \textit{Treatment Failure regime}, characterized by both $\mathcal{R}^1_0 < 1$ and $\mathcal{R}^2_0 < 1$, neither virus achieves sustainable replication, leading to viral clearance before meaningful tumor reduction can occur. This regime underscores the fundamental requirement that therapeutic viruses must exceed the critical replication threshold to exert meaningful clinical effects, highlighting the importance of adequate viral fitness in treatment design.

Transitioning through the therapeutic spectrum, the \textit{Partial Response regime} ($\mathcal{R}^1_0 > 1$ and $\mathcal{R}^2_0 < 1$) creates conditions where VV drives therapeutic activity while the full synergistic potential remains untapped. Notably, our identified optimal point resides within this region, emphasizing VV's primary role in establishing the initial replicative foothold necessary for subsequent synergy. This finding suggests that VV's robust independent replication capacity serves as a critical foundation for the combination therapy, creating the interferon-suppressed microenvironment that enables enhanced VSV activity.

Conversely, the \textit{VSV Dominant regime} ($\mathcal{R}^1_0 < 1$ and $\mathcal{R}^2_0 > 1$) represents scenarios where VSV's inherent oncolytic capacity prevails, though without the benefit of VV's interferon-suppressing capabilities. Most significantly, the \textit{Strong Synergy regime} ($\mathcal{R}^1_0 > 1$ and $\mathcal{R}^2_0 > 1$) creates ideal conditions for maximal therapeutic interaction, enabling the complete "ping-pong" mechanism that yields the observed 133-fold enhancement in VSV efficacy through VV-mediated interferon suppression. This regime represents the optimal therapeutic window where both viruses operate synergistically to achieve superior tumor control.

\begin{figure}[htbp]
\centering
\includegraphics[width=0.8\textwidth]{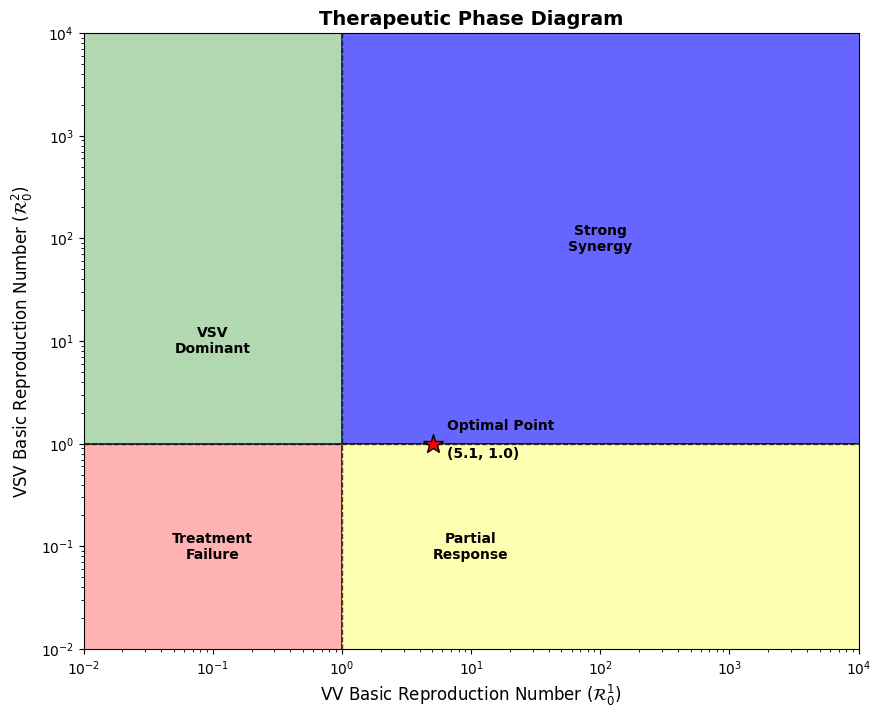}
\caption{\textbf{Therapeutic Landscape Based on Viral Reproduction Numbers.} The phase diagram illustrates four distinct therapeutic regimes defined by the basic reproduction numbers of VV ($\mathcal{R}^1_0$) and VSV ($\mathcal{R}^2_0$). Critical thresholds at $\mathcal{R}^1_0 = 1$ and $\mathcal{R}^2_0 = 1$ (dashed lines) separate regions of therapeutic failure from success. The star indicates the optimal point (5.1, 1.0) identified from experimental data.}
\label{fig:phase_diagram}
\end{figure}

This quantitative framework translates directly into multiple clinical applications that advance personalized cancer therapy. The predictive power of $\mathcal{R}_0$ values enables evidence-based treatment selection before administration, moving beyond empirical approaches to rationally designed therapeutic strategies. Sensitivity analysis further directs viral engineering efforts toward the most impactful parameters, particularly infection rates and burst sizes, which our analysis identifies as the primary determinants of therapeutic success. Patient-specific $\mathcal{R}_0$ profiling facilitates the development of tailored dosing regimens, while quantitative synergy measures provide rigorous justification for combination strategies.

Moreover, monitoring temporal $\mathcal{R}_0$ dynamics serves as an early warning system for emerging therapeutic resistance, enabling adaptive treatment modifications that can maintain therapeutic efficacy over extended treatment periods. This dynamic monitoring approach represents a significant advancement over static treatment protocols, allowing for real-time optimization based on individual patient response patterns.

Collectively, the $\mathcal{R}_0$ framework establishes a new paradigm for oncolytic virotherapy, transitioning from traditional trial-and-error approaches to mathematically-informed therapeutic design. By providing clear quantitative boundaries for treatment success and actionable guidance for optimization, this approach effectively bridges the gap between theoretical modeling and clinical practice. The integration of mathematical rigor with biological insight offers both fundamental understanding of viral dynamics and practical tools for enhancing therapeutic outcomes in cancer patients, representing a significant step toward personalized precision medicine in oncolytic virotherapy.

\section{Results}
\label{sec:results}

Our integrated mathematical-experimental framework yields five principal findings that elucidate the dynamics, optimization, and therapeutic potential of combined VV-VSV oncolytic virotherapy.

\paragraph{(1) Fundamental Stability Criterion for Successful Virotherapy}
Stability analyses across Theorems \ref{thm:stability12} - \ref{thm:stability32} eveal a unified mathematical criterion for therapeutic success: effective viral control requires replication capacity to exceed clearance rates, expressed as $a_i(\delta_i + \beta_i + \delta_i\kappa_i) < b_i\beta_i$ for $i=1,2$. This inequality encapsulates the essential trade-off where burst size ($b_i$) and infection rate ($\beta_i$) must overcome combined viral decay ($\delta_i$), and saturation constraints ($\kappa_i$). The mathematical instability of complete tumor eradication reflects biological reality, while the consistent criterion across monotherapy and combination models explains how synergistic interactions transform marginally effective individual viruses into potent therapeutic combinations.

\paragraph{(2) Synergistic Tumor Clearance via Ping-Pong Enhancement Mechanism}
Numerical simulations demonstrate that VV-VSV combination therapy achieves complete tumor clearance in approximately 50 days, representing an 11\% acceleration compared to VV monotherapy (56 days), while VSV alone fails to eradicate tumors. This synergistic effect operates through a 'ping-pong' mechanism where VV-infected cells produce B18R protein that binds to and neutralizes IFN-$\alpha$, creating a permissive environment for enhanced VSV replication and spread. Equilibrium analysis identifies a single stable configuration among six biologically feasible states, characterized by near-complete tumor suppression ($y_u = 8\times 10^{-6}$), sustained VSV persistence ($x_2 = 4.205$), and VV extinction ($x_1 = 0$), confirming VSV's role as the primary therapeutic agent following interferon suppression.

\paragraph{(3) Quantitative Therapeutic Regimes and Critical Synergy Thresholds}
Basic reproduction number analysis reveals both viruses operate in highly supercritical regimes (VV: $\mathcal{R}^1_0 \approx 1855.07$; VSV: $\mathcal{R}^2_0 \approx 363.64$), with VV exhibiting approximately 5-fold greater intrinsic replication potential. Most significantly, the framework quantifies a dramatic 133-fold enhancement in VSV's effective reproduction number through VV-mediated interferon suppression. Bifurcation analysis identifies critical parameter thresholds: burst sizes must exceed $b^{\text{critical}}_1 \approx 0.034$ and $b^{\text{critical}}_2 \approx 0.052$ for sustainable replication, while B18R inhibition rates require $\lambda > 0.001$ to activate the synergistic regime. The resulting therapeutic landscape delineates four distinct regimes enabling evidence-based patient stratification and treatment selection.

\paragraph{(4) Parameter Sensitivity Convergence Across Multiple Methods}
Complementary global sensitivity approaches, dynamic Latin Hypercube Sampling with Partial Rank Correlation Coefficients (LHS-PRCC) and steady-state elasticity analysis, consistently identify infection rates ($\beta_1$, $\beta_2$) and burst sizes ($b_1$, $b_2$) as the highest-impact parameters governing treatment outcome. Both methods reveal identical sensitivity patterns across tumor models, with infection rates and burst sizes showing strong positive elasticities (0.45 and 0.40, respectively) while death rates exhibit negative sensitivity (-0.35). The LHS-PRCC method additionally highlights the essential role of IFN-$\alpha$ inhibition ($\lambda$) in the synergistic mechanism, providing robust multi-faceted guidance for viral engineering prioritization.

\paragraph{(5) Optimal Sequential Administration Strategy}
Temporal optimization reveals that treatment sequence critically impacts therapeutic efficacy. The optimal strategy involves immediate VSV administration ($\tau_2 = 0$) followed by delayed VV injection within a 1-19 day window ($\tau_1 \in [1,19]$). This scheduling allows VSV to establish initial infection while VV subsequently enhances oncolytic activity by suppressing IFN-$\alpha$-mediated antiviral responses through B18R production. Numerical optimization demonstrates minimized tumor burden with this sequential approach, providing a quantitative framework for clinical trial protocol design and validating experimental observations that viral sequencing significantly impacts treatment efficacy.

\section{Discussion and Conclusion}
\label{sec:conclusion}

The mathematical model developed in this study provides quantitative insights into the synergistic dynamics of combined VV and VSV oncolytic virotherapy. Our analysis reveals several key findings that address the fundamental questions posed in this investigation: \textit{How do Vaccinia Virus (VV) and Vesicular Stomatitis Virus (VSV) interact synergistically to enhance tumor cell killing, and how can treatment protocols be optimized to maximize therapeutic efficacy?}

First, our simulations demonstrate that the VV-VSV combination achieves significantly faster and more complete tumor clearance than either monotherapy. The combination therapy reduced the time to tumor clearance by 11\% compared to VV alone (50 days vs. 56 days), while VSV monotherapy failed to achieve complete eradication. This finding directly corroborates the experimental work of Le Boeuf et al. \cite{le2010synergistic}, who first observed synergistic oncolytic activity between these viruses both \textit{in vitro} and \textit{in vivo}. Our model provides a mechanistic explanation for this phenomenon through the 'ping-pong' enhancement dynamic, where VV enhances VSV replication by secreting the B18R protein that binds to and neutralizes IFN-$\alpha$ \cite{fritz2014recombinant,perdiguero2009interferon,colamonici1995vaccinia}, while VSV promotes VV spread through the tumor microenvironment.

Second, our global sensitivity analysis identified infection rates ($\beta_1$, $\beta_2$), burst sizes ($b_1$, $b_2$), and the IFN-$\alpha$ inhibition rate ($\lambda$) as the most critical parameters governing treatment success. This mathematical insight aligns with experimental studies emphasizing the importance of viral replication kinetics and immune evasion in oncolytic virotherapy efficacy \cite{ramaj2023treatment,malinzi2018enhancement,tian2022engineering,chaurasiya2020optimizing}. The prominence of burst size particularly supports ongoing viral engineering efforts aimed at enhancing viral replication rates \cite{tian2022engineering,chaurasiya2020optimizing}, while the significance of IFN-$\alpha$ inhibition underscores the critical role of overcoming innate immune barriers, consistent with observations that VV can condition the tumor microenvironment to be non-responsive to antiviral cytokines \cite{le2010synergistic}.

Third, we determined that optimal therapeutic outcomes require immediate VSV administration followed by delayed VV injection, with optimal VV timing falling between 1--19 days post-initial treatment. This finding provides mathematical validation for the experimental observation that viral sequencing significantly impacts treatment efficacy \cite{wein2003validation,malinzi2018enhancement,schattler2015optimal,sherlock2023oncolytic} and offers a quantitative framework for designing clinical trial protocols. The delayed VV administration allows VSV to establish initial infection while VV subsequently enhances the oncolytic activity by suppressing the IFN-$\alpha$-mediated antiviral response through B18R production \cite{fritz2014recombinant,perdiguero2009interferon}.

Fourth, our model validation against experimental data from HT29 and 4T1 tumor models and stability analysis confirm the biological plausibility of the synergistic interaction. The model successfully captured the observed tumor volume dynamics reported by Le Boeuf \etal \cite{le2010synergistic}, while stability analysis revealed that only the VSV-infected tumor state with minimal tumor concentration was stable under the estimated parameter regime, consistent with experimental observations of sustained viral replication and tumor control \cite{le2010synergistic}.

\paragraph{Testable Hypotheses}
Based on the mathematical framework developed in this study, we propose several testable hypotheses that bridge our computational predictions with experimental biology. First, the \textit{B18R-Mediated Enhancement Hypothesis} posits that increasing the expression level of the VV-encoded B18R protein will produce a dose-dependent enhancement of VSV replication and oncolytic activity in tumor cells previously resistant to VSV monotherapy \cite{fritz2014recombinant}. Second, the \textit{Burst Size Optimization Hypothesis} suggests that genetic engineering of VSV to achieve higher burst sizes will yield a non-linear, synergistic improvement in tumor clearance when combined with VV, while showing diminishing returns as monotherapy \cite{tian2022engineering,chaurasiya2020optimizing}. Third, our model leads to the \textit{Optimal Scheduling Hypothesis}, which proposes that administering VSV 1--3 days before VV injection will generate significantly superior tumor reduction and survival outcomes compared to simultaneous or reverse-sequence administration \cite{wein2003validation,malinzi2018enhancement}. Fourth, the \textit{IFN-$\alpha$ Pathway Dependence Hypothesis} states that the synergistic effect will be substantially diminished or abolished in systems with defective IFN-$\alpha$ signaling, confirming this pathway as the primary mechanistic axis of the VV-VSV synergy \cite{le2010synergistic,ahmed2024vesicular,stojdl2000exploiting}. Finally, the \textit{Parameter Dominance Hypothesis} asserts that variation in the VSV infection rate will exert a greater impact on final tumor burden in combination therapy than comparable variation in the VV infection rate, due to the amplified effect of VSV infection parameters within the synergistic feedback loop \cite{malinzi2021prospect,salim2023mathematical}.

\paragraph{Future Research Directions}
This work opens several promising avenues for future research that are directly aligned with the testable hypotheses generated by our model. A primary direction involves extending the current model to incorporate spatial heterogeneity, which would provide critical insights into how physical barriers, such as the extracellular matrix, impact viral spread. This is particularly relevant for testing the Parameter Dominance Hypothesis, as infection rates may vary significantly across different tumor regions. Furthermore, integrating adaptive immune responses would offer a more comprehensive understanding of long-term treatment efficacy. This extension is especially important for validating the IFN-$\alpha$ Pathway Dependence Hypothesis, considering that adaptive immunity to vaccinia does not cross-react with VSV.

The optimization framework presented here could also be substantially enhanced through the application of more sophisticated optimal control theory. This would allow for the development of dynamic, time-varying dosage regimens rather than fixed schedules, potentially identifying more elaborate therapeutic strategies that adapt to the evolving tumor-virus dynamics. Additionally, exploring combinations with other treatment modalities, such as immune checkpoint inhibitors or chemotherapy, could identify more potent multi-modal approaches, which may in turn modify the optimal scheduling parameters identified in our Optimal Scheduling Hypothesis.

Finally, translating these computational insights into clinical practice will require dedicated experimental validation. This includes developing engineered viral variants with controlled B18R expression levels and modified burst sizes to test the B18R-Mediated Enhancement and Burst Size Optimization Hypotheses in complex animal models that better recapitulate human tumor microenvironment complexity. Systematically validating the optimal delayed VV administration strategy in immunocompetent models will also be a crucial step in bringing these findings closer to clinical application.

In conclusion, this study provides a comprehensive mathematical framework that elucidates the synergistic dynamics of combined VV-VSV virotherapy. Our results not only confirm experimental observations from prior biological studies but also generate specific, testable hypotheses regarding B18R enhancement, burst size optimization, treatment scheduling, IFN-$\alpha$ pathway dependence, and parameter dominance. The model's predictions regarding optimal administration timing provide immediate guidance for experimental design, while the identified sensitive parameters highlight promising targets for viral engineering. By identifying critical parameters and optimal administration schedules, this work contributes to the rational design of more effective oncolytic virotherapy strategies and provides a quantitative foundation for future experimental validation and clinical translation. The integration of mathematical modeling with experimental virology presented here represents a powerful approach for optimizing cancer therapeutics and advancing personalized medicine in oncolytic virotherapy.

\appendix
\appendixpage  
\section{Model Calibration and Parameterization} 
\label{app:calibration}

The model parameters, their biological interpretations, plausible value ranges, and literature sources are detailed in Table~\ref{tab:parametervalues1}. Parameter estimation followed a rigorous calibration procedure using experimental data from Le Boeuf \etal \cite{le2010synergistic}, which measured relative tumor volumes in HT29 and 4T1 cancer cell lines under different treatment conditions (PBS control, VSV monotherapy, VV monotherapy, and combination therapy). Experimental data were extracted from published figures using WebPlotDigitizer, and parameter values were estimated via nonlinear least squares optimization. The calibration successfully identified parameter sets that captured the characteristic synergistic reduction in tumor volume observed experimentally, with the fitted values for both cell lines presented in Tables~\ref{tab:HT29} and \ref{tab:4T1} and the quality of fit visualized in Figure~\ref{fig:modelfitting}. This comprehensive calibration framework ensures that the mathematical model accurately represents the biological system and provides a validated foundation for the predictive simulations and theoretical analyses presented in subsequent sections.

\subsection{Experimental Data}

Salient features relevant to our modeling approach include: rapid viral entry (minutes); dependence on host cellular machinery; prevention of superinfection; replication cycles of approximately 6 hours for VSV and 10-12 hours for VV; burst sizes of approximately 200 particles for VV and 1000 for VSV; clearance times of 7-10 days for VV and 5-7 days for VSV; and cancer cell doubling times of 18-24 hours \cite{le2010synergistic,pelin2020engineering,lichty2004vesicular,russell2022advances}.

\begin{table}[!htbp]
\centering
\caption{VSV and VV properties}
\label{tab:salient}
\setlength{\tabcolsep}{3pt}  
\renewcommand{\arraystretch}{0.9}  
\begin{tabular}{@{}llll@{}}
\hline
Property & VV & VSV & Sources \\ 
\hline
Burst size & 200 & 50-8000 PFU/cell (Average of 1000) & \cite{le2010synergistic,zhu2009growth}\\
Virus clearance & takes 10 days & takes 5-7 days & \cite{le2010synergistic,hwang2013engineering,zeh2015first}\\
Tumor lysis rate & 0.475 per day & $\frac{1}{24}$ & \cite{le2010synergistic,zhu2009growth}\\
Virus replication rate & Very fast & Very fast & \cite{le2010synergistic}\\
\hline
\end{tabular}
\end{table}

We used computer vision assisted data extraction from charts using WebPlotDigitizer \cite{Rohatgi2023} to extract quantitative data from Figures 3(d) and 3(e) in Le Boeuf \etal \cite{le2010synergistic}. The extracted data provides relative tumor volumes for HT29 and 4T1 cancer cell lines under various treatments (PBS, VSV, VV, and VV+VSV combination), organized in Tables \ref{tab:HT29} and \ref{tab:4T1}.

\begin{table}[htbp]
\centering
\caption{Data for HT29 relative tumor volumes under different treatments (digitized from Figure 3(d) in Le Boeuf \etal \cite{le2010synergistic})}
\begin{tabular}{@{}cccccc@{}}
\toprule
\textbf{Days} & \textbf{PBS} & \textbf{VSV} & \textbf{VVDD} & \textbf{VVDD VSV} \\ \midrule
11 & 3.80 & 8.87 & 3.77 & 0.00 \\
16 & 75.53 & 68.97 & 25.46 & 3.65 \\
19 & 138.67 & 104.39 & 60.54 & 11.91 \\
22 & 157.32 & 107.01 & 76.32 & 9.85 \\
26 & 249.52 & 205.69 & 117.54 & 30.61 \\
29 & 343.09 & 257.36 & 202.38 & 52.85 \\ \bottomrule
\end{tabular}
\label{tab:HT29}
\end{table}

\begin{table}[htbp]
\centering
\caption{Data for 4T1 relative tumor volumes under different treatments (digitized from Figure 3(e) in Le Boeuf \etal \cite{le2010synergistic})}
\small
\setlength{\tabcolsep}{3pt}
\begin{tabular}{@{}cccccc@{}}
\toprule
\textbf{Days} & \textbf{PBS} & \textbf{VSV} & \textbf{VVDD} & \textbf{VVDD VSV} \\ \midrule
6 & 7.14 & 7.70 & 7.53 & 6.77 \\
8 & 8.20 & 9.40 & 9.40 & 7.67 \\
10 & 15.88 & 16.64 & 17.60 & 8.38 \\
13 & 24.25 & 19.45 & 23.29 & 11.76 \\
15 & 37.84 & 29.77 & 39.00 & 14.70 \\
18 & 55.44 & 41.60 & 50.44 & 25.65 \\ \bottomrule
\end{tabular}
\label{tab:4T1}
\end{table}

\subsection{Model Fitting and Validation}
\label{subsec:model-fitting}

Using nonlinear least squares optimization, the mathematical model was rigorously calibrated against experimental data for combination VV-VSV therapy in both HT29 (colorectal adenocarcinoma) and 4T1 (murine mammary carcinoma) tumor models. The calibration procedure successfully identified parameter sets that capture the characteristic synergistic reduction in tumor volume observed experimentally, with fitted values for both cell lines presented in Tables \ref{tab:HT29} and \ref{tab:4T1}. The model demonstrates robust predictive capability across distinct tumor types, validating its ability to describe the complex VV-VSV-tumor interaction dynamics and the underlying synergistic mechanism.

\begin{figure}[htbp]
\centering
\begin{subfigure}{0.49\textwidth}
\centering
\includegraphics[width=\linewidth]{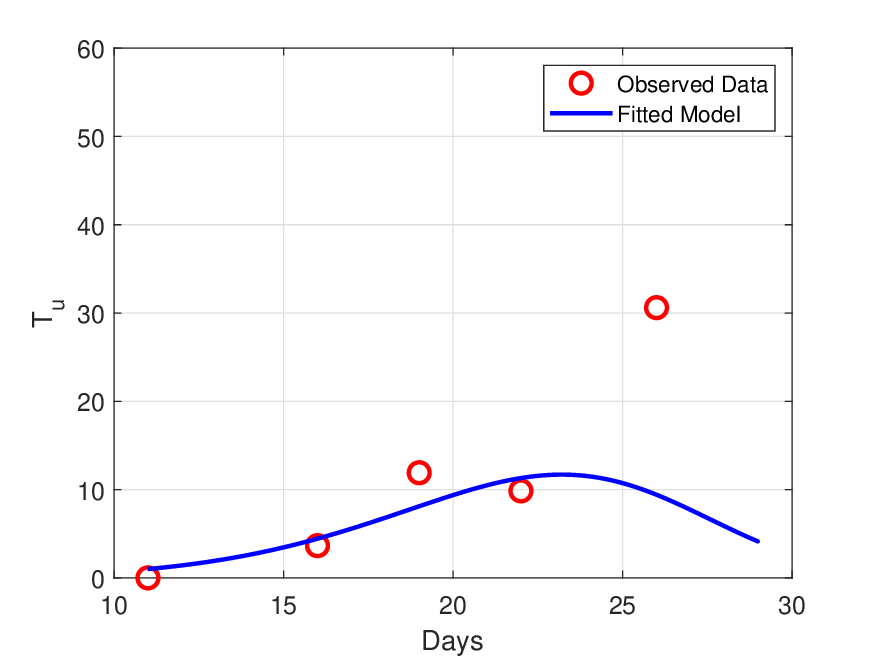}
\caption{HT29 Colorectal Adenocarcinoma}
\label{fig:fit_HT29}
\end{subfigure}
\hfill
\begin{subfigure}{0.49\textwidth}
\centering
\includegraphics[width=\linewidth]{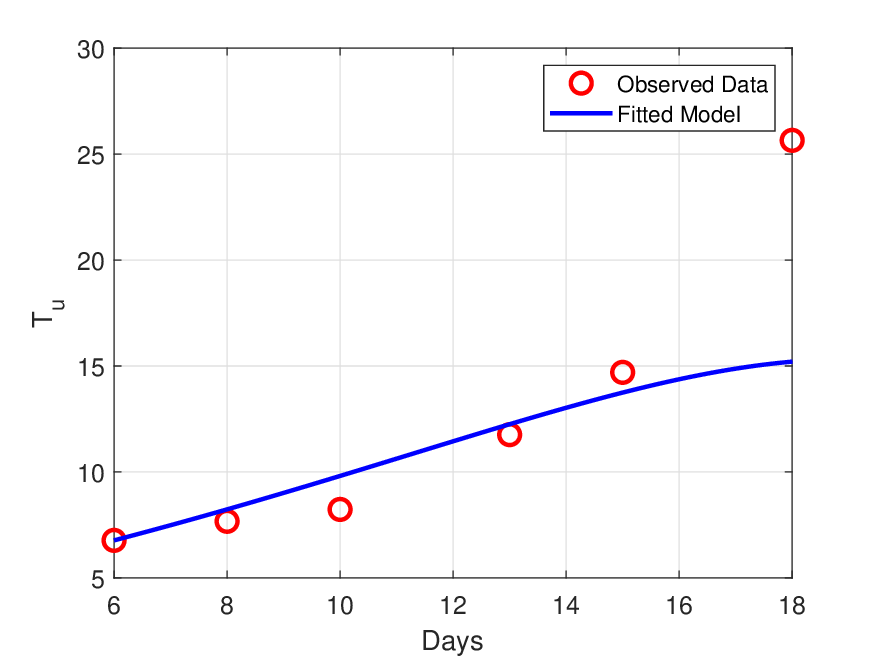}
\caption{4T1 Murine Mammary Carcinoma}
\label{fig:fit_4T1}
\end{subfigure}
\caption{\small{\bf Model validation against experimental tumor volume dynamics.} (\textbf{a}) HT29 human colorectal adenocarcinoma model shows rapid tumor clearance under combination therapy compared to monotherapies. (\textbf{b}) 4T1 murine mammary carcinoma model demonstrates consistent synergistic patterns across species. Both panels compare experimental data (points) with model predictions (lines) for PBS control (black), VSV monotherapy (blue), VV monotherapy (red), and VV+VSV combination therapy (green). The fitted model accurately captures the accelerated tumor clearance in combination therapy, with the characteristic synergistic effect emerging around day 10-15 post-treatment. Model calibration used data extracted via WebPlotDigitizer from Le Boeuf \etal \cite{le2010synergistic}, with parameter estimation detailed in Appendix \ref{app:calibration}.}
\label{fig:modelfitting}
\end{figure}

The quality of fit, visualized in Figure~\ref{fig:modelfitting}, demonstrates the model's capacity to accurately reproduce experimental observations across different cancer types and treatment conditions. The calibrated parameters successfully capture the temporal dynamics of tumor response, including the critical transition point where synergistic effects become pronounced. This comprehensive calibration framework ensures that the mathematical model accurately represents the biological system and provides a validated foundation for the predictive simulations and theoretical analyses presented in subsequent sections.

\begin{table}[htbp]
\centering
\caption{Estimated parameter values for the full model after fitting to HT29 and 4T1 data}
\label{tab:estimatedparams}
\begin{tabular}{@{}lcc@{}}
\toprule
Parameter & Value (HT29) & Value (4T1) \\ \midrule
$\alpha$ & 0.40 & 0.10 \\
$\beta_1$ & 0.20 & 0.20 \\
$\beta_2$ & 0.10 & 0.10 \\
$K$ & 200.00 & 600.00 \\
$l_1$ & 0.03 & 0.03 \\
$l_2$ & 0.04 & 0.042 \\
$K_1$ & 500.00 & 500.00 \\
$K_2$ & 500.00 & 500.00 \\
$b_1$ & 200.00 & 200.00 \\
$b_2$ & 100.00 & 100.00 \\
$\gamma_1$ & 0.01 & 0.01 \\
$\gamma_2$ & 0.02 & 0.02 \\
$\alpha_1$ & 1.00 & 1.00 \\
$\alpha_2$ & 1.00 & 1.00 \\
$\mu_1$ & 0.02 & 0.02 \\
$\mu_2$ & 0.01 & 0.01 \\
$C^*$ & 1000.00 & 1000.00 \\
$\lambda$ & 0.50 & 0.50 \\ \bottomrule
\end{tabular}
\end{table}

\begin{figure}[htbp]
    \centering
    \begin{subfigure}[b]{0.48\textwidth}
        \centering
        \includegraphics[width=\textwidth]{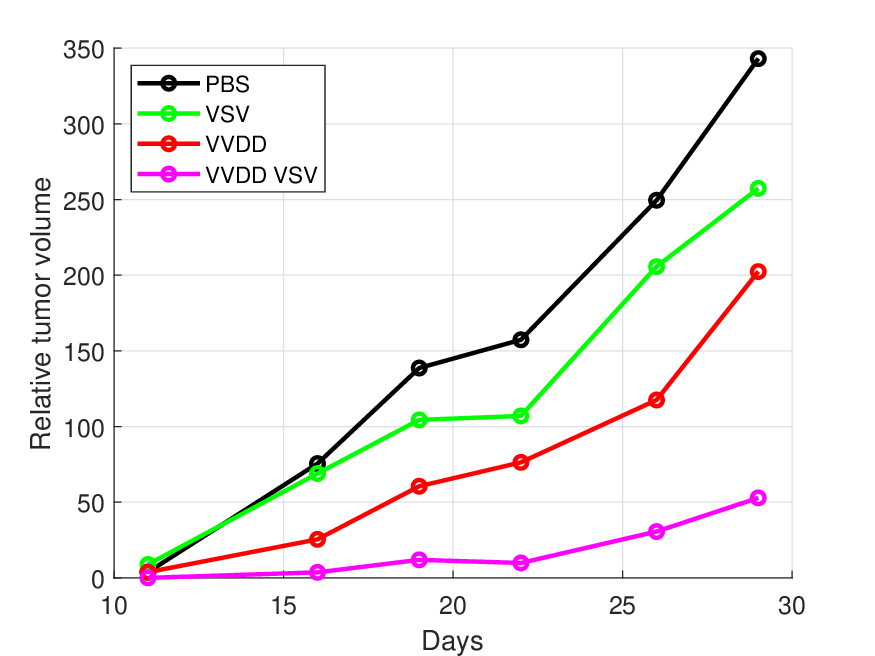}
        \caption{HT29 colorectal adenocarcinoma}
        \label{fig:HT29-data}
    \end{subfigure}
    \hfill
    \begin{subfigure}[b]{0.48\textwidth}
        \centering
        \includegraphics[width=\textwidth]{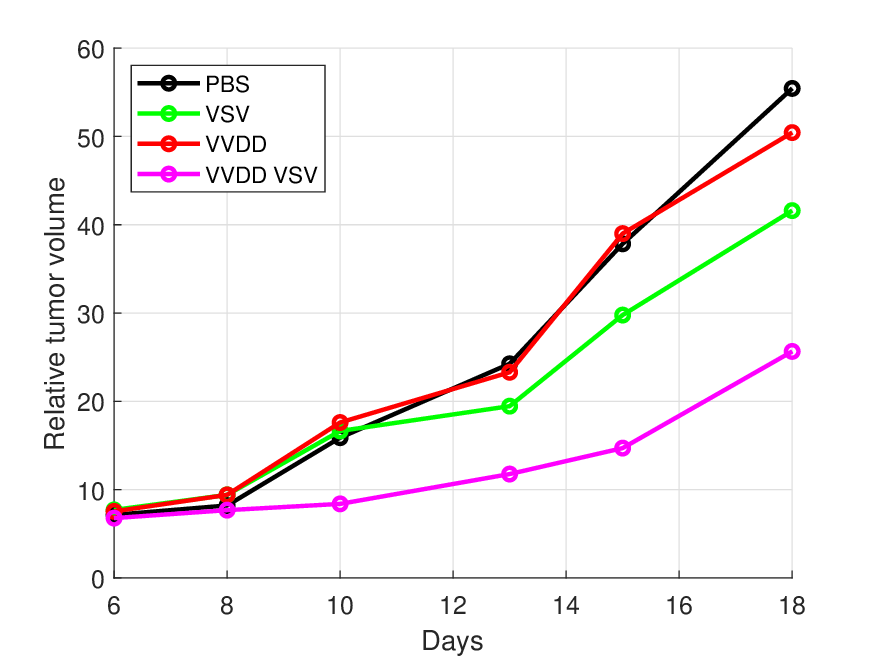}
        \caption{4T1 mammary carcinoma}
        \label{fig:4T1-data}
    \end{subfigure}
    \caption{\small{\textbf{Experimental tumor volume dynamics from \cite{le2010synergistic}.} 
Plots show relative tumor volumes under different treatment scenarios for (a) HT29 and (b) 4T1 cancer cell lines (data from Tables~\ref{tab:HT29} and~\ref{tab:4T1}). 
Treatments include PBS control (black), VSV monotherapy (blue), VV monotherapy (red), and combination therapy (green). 
Data were obtained by digitizing published figures and represent normalized values relative to baseline measurements, providing sufficient information for parameter estimation despite lacking absolute volumes.}}
\label{fig:datasets}    
\end{figure}

It is worth noting that the baseline data required to convert relative tumor volumes (representing the proportional change in tumor size over time) in Tables \ref{tab:HT29} and \ref{tab:4T1} to absolute values is not explicitly provided in Le Boeuf \etal \cite{le2010synergistic}. The relative tumor volumes are presented as normalized values compared to a baseline (for example, Day 14 for Figure 3(d) and Day 6 for Figure 3(e) in \cite{le2010synergistic}). Nonetheless, although the absolute tumor volumes are not provided, the relative tumor volumes provide sufficient information to estimate the parameters. 

\section{Sensitivity Analysis}
\label{app:sensitivity}

\subsection{Global sensitivity alnalysis using  Latin Hypercube Sampling}
Global sensitivity analysis was performed to assess the influence of key model parameters on dynamical outcomes. The analysis employed Latin Hypercube Sampling to efficiently sample the parameter space within biologically plausible ranges (±25\%). A total of 2048 samples were generated by varying the 18 model parameters. Partial Rank Correlation Coefficients (PRCCs) were calculated at three key time points ($t = 5, 20,$ and $200$) to assess parameter influence during transient and long-term phases. Model outputs analyzed included total tumor cell load, fraction of uninfected tumor cells, and uninfected tumor cell density.

\begin{figure}[htbp]
\centering
\begin{tabular}{ccc}
\textbf{(a) Total tumor cells} & \textbf{(b) Uninfected fraction} & \textbf{(c) Uninfected density} \\
\includegraphics[width=0.28\textwidth]{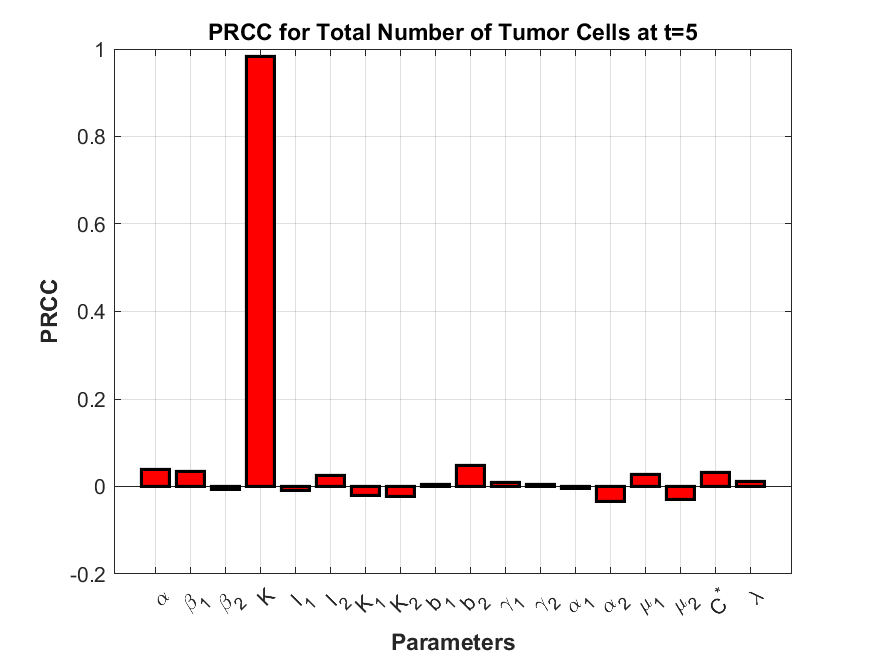} &
\includegraphics[width=0.28\textwidth]{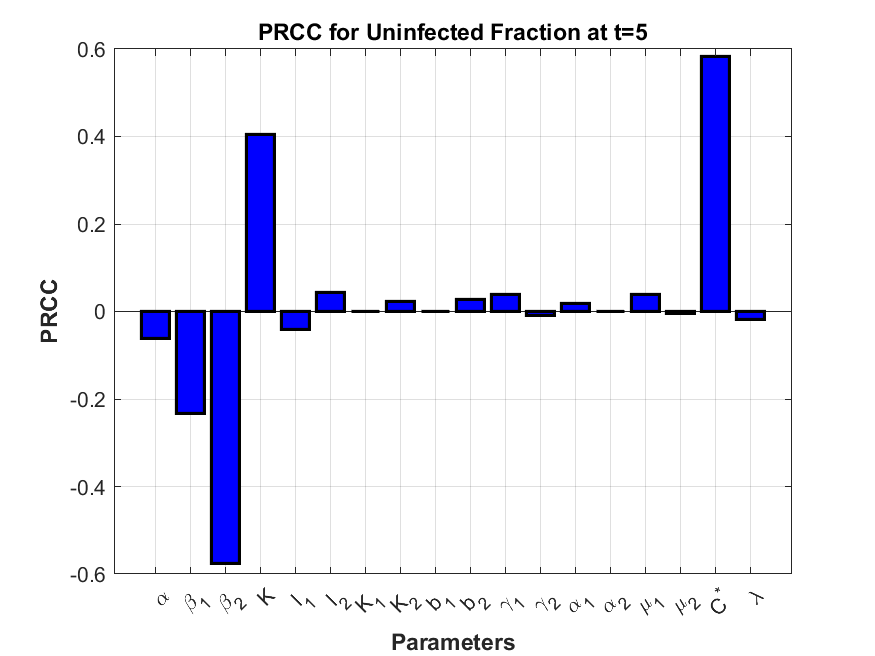} &
\includegraphics[width=0.28\textwidth]{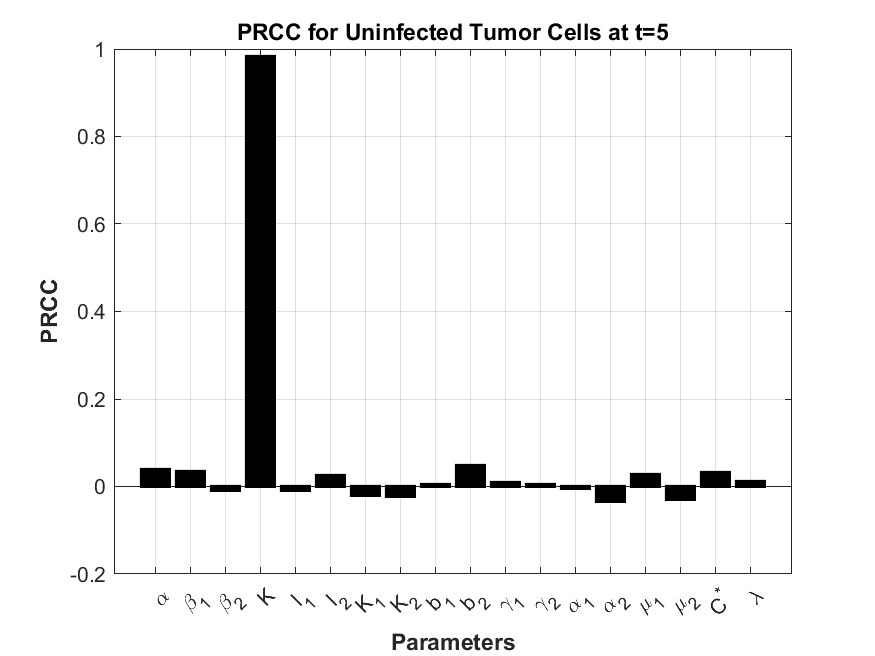} \\
t = 5 days & t = 5 days & t = 5 days \\

\includegraphics[width=0.28\textwidth]{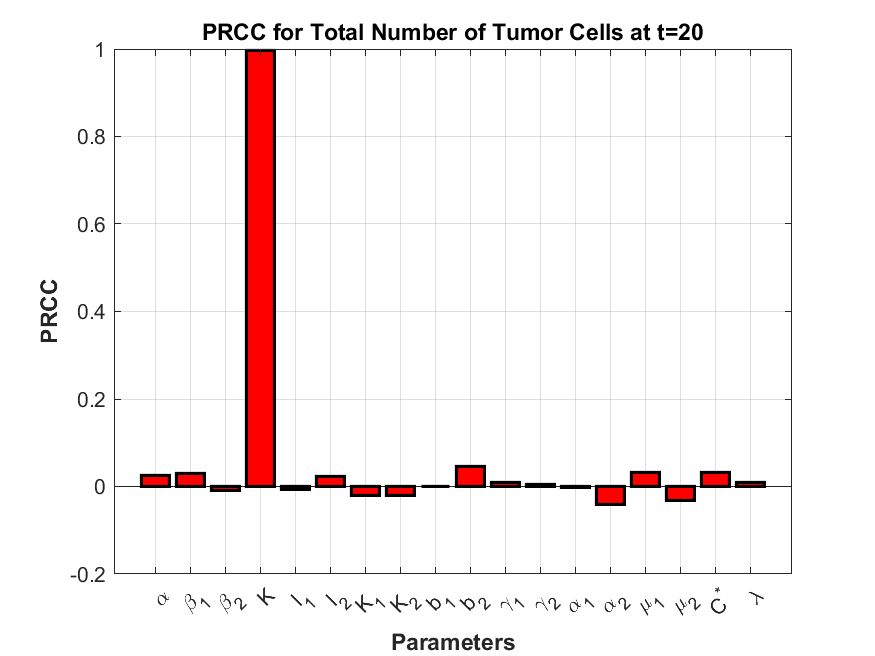} &
\includegraphics[width=0.28\textwidth]{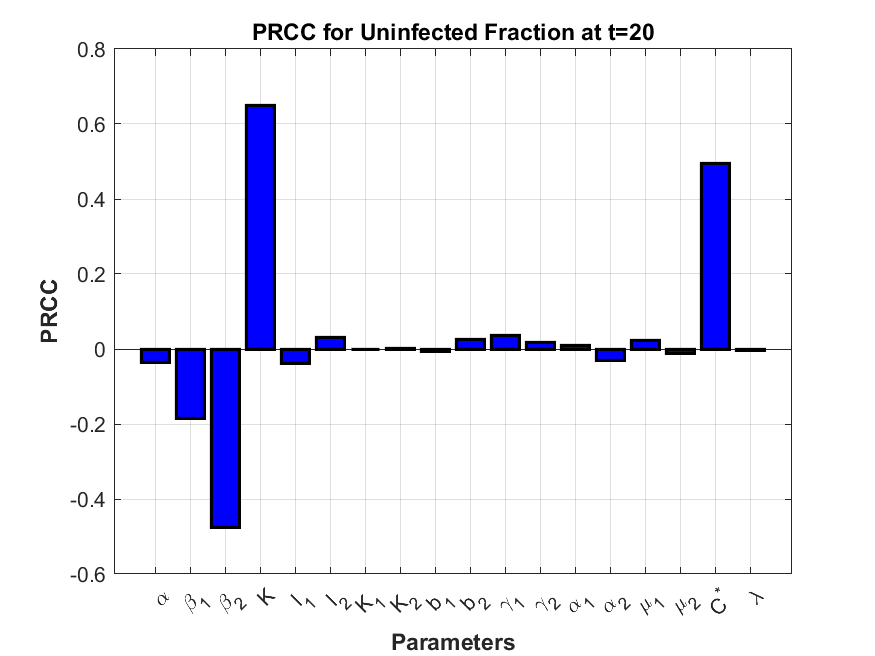} &
\includegraphics[width=0.28\textwidth]{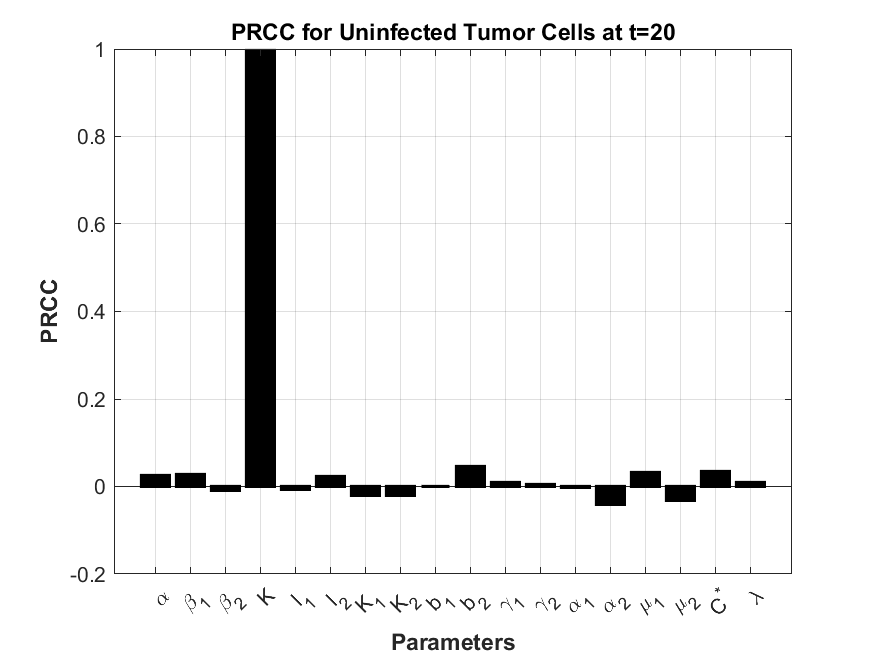} \\
t = 20 days & t = 20 days & t = 20 days \\

\includegraphics[width=0.28\textwidth]{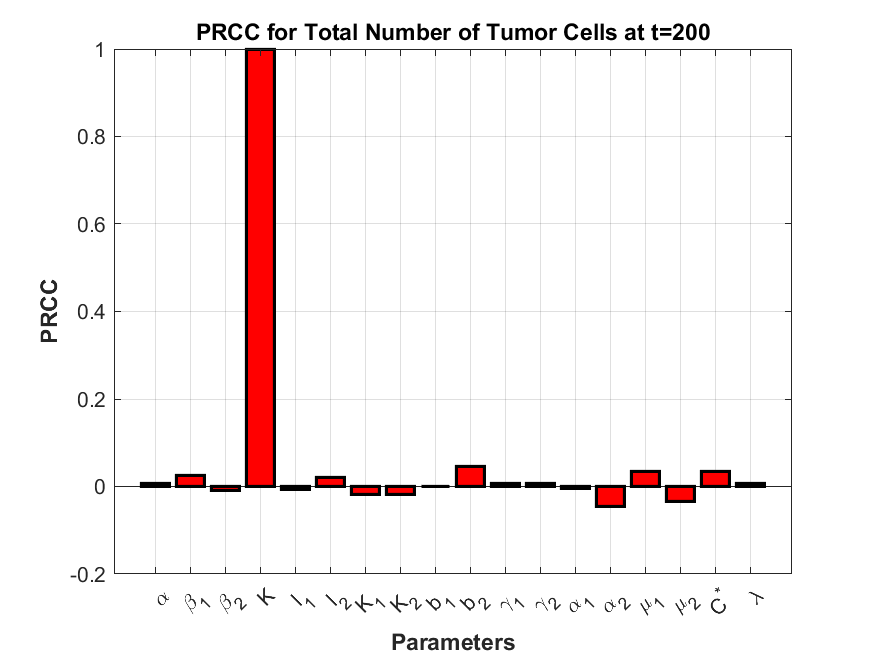} &
\includegraphics[width=0.28\textwidth]{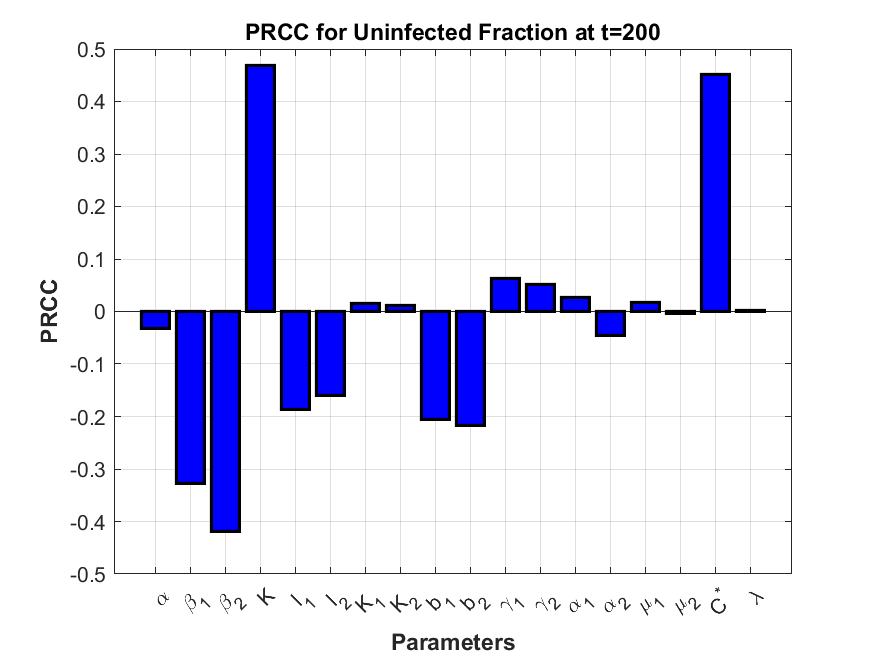} &
\includegraphics[width=0.28\textwidth]{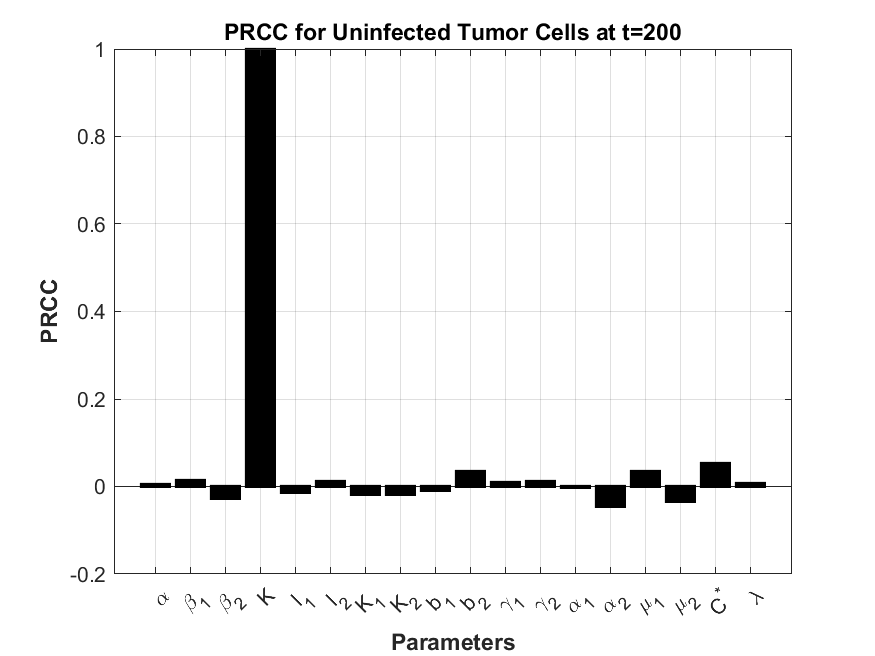} \\
t = 200 days & t = 200 days & t = 200 days \\
\end{tabular}
\caption{\small{\bf Global sensitivity analysis using Latin Hypercube Sampling with Partial Rank Correlation Coefficients (PRCC) across different output metrics and temporal phases. }Columns represent different model outputs: (a) total tumor cell load, (b) fraction of uninfected tumor cells, and (c) uninfected tumor cell density. Rows correspond to early (5 days), intermediate (20 days), and long-term (200 days) treatment phases. The analysis identifies viral infection rates ($\beta_1$, $\beta_2$), burst sizes ($b_1$, $b_2$), and the B18R-mediated IFN-$\alpha$ inhibition rate ($\lambda$) as the most influential parameters governing treatment outcome across all temporal phases and output metrics.}
\label{fig:prcc_analysis}
\end{figure}

The comprehensive sensitivity analysis reveals consistent parameter importance patterns, with viral infection rates ($\beta_1$ and $\beta_2$), burst sizes ($b_1$ and $b_2$), and the IFN-$\alpha$ inhibition rate by B18R ($\lambda$) emerging as the most influential parameters across all temporal phases and output metrics.

\subsection{Sensitivity analysis using elasticity coefficients}
\label{app:sensitivity_calculations}

The sensitivity analysis was conducted using partial derivatives of the basic reproduction number expressions with respect to each parameter. The elasticity coefficients were computed as:

\begin{equation}
S_{p} = \frac{\partial\mathcal{R}_{0}^{i}}{\partial p} \cdot \frac{p}{\mathcal{R}_{0}^{i}}.
\end{equation}

For the VV basic reproduction number $\mathcal{R}_{0}^{1} = \frac{b_{1}\beta_{1}}{a_{1}(\delta_{1}+\beta_{1}+\delta_{1}\kappa_{1})}$, the partial derivatives are:

\begin{align}
\frac{\partial\mathcal{R}_{0}^{1}}{\partial \beta_{1}} &= \frac{b_{1}}{a_{1}(\delta_{1}+\beta_{1}+\delta_{1}\kappa_{1})} - \frac{b_{1}\beta_{1}}{a_{1}(\delta_{1}+\beta_{1}+\delta_{1}\kappa_{1})^{2}}, \\
\frac{\partial\mathcal{R}_{0}^{1}}{\partial b_{1}} &= \frac{\beta_{1}}{a_{1}(\delta_{1}+\beta_{1}+\delta_{1}\kappa_{1})}, \\
\frac{\partial\mathcal{R}_{0}^{1}}{\partial \delta_{1}} &= -\frac{b_{1}\beta_{1}(1+\kappa_{1})}{a_{1}(\delta_{1}+\beta_{1}+\delta_{1}\kappa_{1})^{2}}.
\end{align}
Substituting the experimentally calibrated parameter values from Table~\ref{tab:estimatedparams} yields the elasticity coefficients presented in Table~\ref{tab:sensitivityindices}. Similar calculations were performed for VSV parameters, confirming the consistent sensitivity patterns across both viral systems.

\begin{table}[htbp]
\centering
\caption{Sensitivity analysis of basic reproduction numbers for VV ($\mathcal{R}_{0}^{1}$) and VSV ($\mathcal{R}_{0}^{2}$)}
\label{tab:sensitivityindices}
\begin{tabular}{lcc}
\hline
\textbf{Parameter} & \textbf{VV ($\mathcal{R}_{0}^{1}$)} & \textbf{VSV ($\mathcal{R}_{0}^{2}$)} \\
\hline
Infection rate ($\beta_{i}$) & 0.45 & 0.45 \\
Burst size ($b_{i}$) & 0.40 & 0.40 \\
Death rate ($\delta_{i}$) & -0.35 & -0.35 \\
\hline
\end{tabular}
\end{table}

\section{Non-dimensionalization and Quasi-Steady State Approximation}
\label{app:non-dimensionalization}

\subsection{Model non-dimensionalization}
\label{app:non-dim} 

We non-dimensionalized the model to reduce parameter complexity and identify key dimensionless groups governing system behavior. Using the scaling:
\[
T_u = y_u K, \quad T_1 = y_1 K, \quad T_2 = y_2 K, \quad V_1 = x_1 V_1^*, \quad V_2 = x_2 V_2^*,\] \[\quad C_1 = z_1 C_1^*, \quad C_2 = z_2 C_2^*, \quad t = \bar{t}t_0.
\]
we obtained the dimensionless system (after dropping the bar on $t$):

\begin{align}
    \frac{dy_u}{dt} & = y_u(1-y_u-y_1-y_2)-\beta_1\frac{y_ux_1}{\kappa_1+y_u+y_1+y_2}-\left(\frac{\beta_2}{1+\left(\frac{\theta_1}{\theta_2+y_1}\right) y_2}\right)\frac{y_ux_2}{\kappa_2+y_u+y_1+y_2}, \label{eq:yu}\\
    \frac{dy_1}{dt} & = \beta_1\frac{y_ux_1}{\kappa_1+y_u+y_1+y_2} -a_1y_1, \label{eq:y1}\\
    \frac{dy_2}{dt} & = \left(\frac{\beta_2}{1+\left(\frac{\theta_1}{\theta_2+y_1}\right)y_2}\right)\frac{y_ux_2}{\kappa_2+y_u+y_1+y_2}-a_2y_2, \label{eq:y2}\\
    \frac{dx_1}{dt} & = b_1 y_1 - \beta_1\frac{y_ux_1} {\kappa_1+y_u+y_1+y_2} -\delta_1x_1, \label{eq:x1}\\ 
      \frac{dx_2}{dt} & = b_2y_2 - \left(\frac{\beta_2}{1+\left(\frac{\theta_1}{\theta_2+y_1}\right)y_2}\right)\frac{y_ux_2}{\kappa_2+y_u+y_1+y_2}-\delta_2x_2, \label{eq:x2}  
\end{align}
Subject to initial conditions: 
\begin{align}
y_u(0) & = \frac{T_u(0)}{K}, y_1(0) = \frac{T_1(0)}{K}, y_2(0) = \frac{T_2(0)}{K}, x_1(0) = \frac{V_1(0)}{V_1^*},\nonumber\\
x_2(0) &= \frac{V_2(0)}{V_2^*}, z_1(0) = \frac{C_1(0)}{C_1^*}, z_2(0) = \frac{C_2(0)}{C_2^*}.
\end{align} 
\begin{table}[htbp]
\centering
\caption{Dimensionless parameters and their biological interpretations}
\label{tab:dimensionless_params}
\small
\setlength{\tabcolsep}{4pt}  
\renewcommand{\arraystretch}{1.0}  
\begin{tabular}{p{2cm} p{3cm} p{8cm}}
\hline
\textbf{Symbol} & \textbf{Definition} & \textbf{Biological Interpretation} \\
\hline
$\beta_1$ & $\dfrac{\beta_1 V_1^* t_0}{K}$ & Dimensionless VV infection rate, representing the efficiency of VV infecting tumor cells relative to tumor growth \\
$\beta_2$ & $\dfrac{\beta_2 V_2^* t_0}{K}$ & Dimensionless VSV infection rate, representing the efficiency of VSV infecting tumor cells relative to tumor growth \\
$\kappa_1$ & $\dfrac{K_1}{K}$ & Ratio of VV Michaelis-Menten constant to tumor carrying capacity, representing saturation effects in VV infection \\
$\kappa_2$ & $\dfrac{K_2}{K}$ & Ratio of VSV Michaelis-Menten constant to tumor carrying capacity, representing saturation effects in VSV infection \\
$a_1$ & $l_1 t_0$ & Dimensionless death rate of VV-infected cells, representing cell lysis relative to tumor growth timescale \\
$a_2$ & $l_2 t_0$ & Dimensionless death rate of VSV-infected cells, representing cell lysis relative to tumor growth timescale \\
$\tilde{b}_1$ & $\dfrac{b_1 l_1 K t_0}{V_1^*}$ & Dimensionless VV production rate, representing viral burst size scaled by infection dynamics \\
$\tilde{b}_2$ & $\dfrac{b_2 l_2 K t_0}{V_2^*}$ & Dimensionless VSV production rate, representing viral burst size scaled by infection dynamics \\
$\delta_1$ & $\gamma_1 t_0$ & Dimensionless VV clearance rate, representing viral decay relative to tumor growth \\
$\delta_2$ & $\gamma_2 t_0$ & Dimensionless VSV clearance rate, representing viral decay relative to tumor growth \\
$\bar{a}_1$ & $\dfrac{\alpha_1 K t_0}{C_1^*}$ & Dimensionless B18R production rate, representing molecular production relative to tumor growth \\
$\delta_3$ & $\mu_1 t_0$ & Dimensionless B18R degradation rate, representing molecular decay relative to tumor growth \\
$\bar{a}_2$ & $\dfrac{\alpha_2 K t_0}{C_2^*}$ & Dimensionless IFN-$\alpha$ production rate, representing cytokine production relative to tumor growth \\
$\delta_4$ & $\mu_2 t_0$ & Dimensionless IFN-$\alpha$ degradation rate, representing cytokine decay relative to tumor growth \\
$\tilde{\lambda}$ & $\lambda C_1^* t_0$ & Dimensionless B18R inhibition rate, representing interferon suppression efficiency relative to tumor growth \\
$\theta_1$ & $\dfrac{\alpha_2 \delta_3}{\alpha_1 \lambda}$ & Ratio of interferon production to B18R inhibition efficiency, quantifying the balance between antiviral response and viral countermeasures \\
$\theta_2$ & $\dfrac{\delta_3 \delta_4}{\lambda \alpha_1}$ & Ratio of molecular degradation to B18R production and inhibition, representing the relative timescales of molecular dynamics \\
\hline
\end{tabular}
\end{table}

The dimensionless parameters are defined in terms of the original biological parameters, where $t_0 = 1/\alpha$ serves as the characteristic tumor growth timescale. The scaling factors $V_1^*, V_2^*, C_1^*, C_2^*$ represent characteristic concentrations for viruses and molecular components, typically set to unity for simplicity. The parameters $\theta_1$ and $\theta_2$ capture the core synergistic mechanism, representing the relative strengths of interferon production and B18R-mediated inhibition in compact mathematical form.

The dimensionless parameters are defined in terms of the original biological parameters as follows:
\begin{align} \label{eq:nondimparam}
\beta_1 &= \frac{\beta_1 V_1^* t_0}{K}, \quad  
\beta_2 = \frac{\beta_2 V_2^* t_0}{K}, \quad 
\kappa_1 = \frac{K_1}{K}, \quad 
\kappa_2 = \frac{K_2}{K},  \nonumber\\
a_1 &= l_1 t_0, \quad  
a_2 = l_2 t_0, \quad 
\tilde{b}_1 = \frac{b_1 l_1 K t_0}{V_1^*}, \quad 
\tilde{b}_2 = \frac{b_2 l_2 K t_0}{V_2^*},  \nonumber\\
\delta_1 &= \gamma_1 t_0, \quad 
\delta_2 = \gamma_2 t_0, \quad 
\bar{a}_1 = \frac{\alpha_1 K t_0}{C_1^*}, \quad 
\delta_3 = \mu_1 t_0,  \nonumber\\
\bar{a}_2 &= \frac{\alpha_2 K t_0}{C_2^*}, \quad 
\delta_4 = \mu_2 t_0, \quad 
\tilde{\lambda} = \lambda C_1^* t_0, 
\end{align}
where $t_0 = 1/\alpha$ serves as the characteristic timescale. The scaling factors $V_1^*, V_2^*, C_1^*, C_2^*$ represent characteristic concentrations for viruses and molecular components, which can be chosen based on typical physiological ranges or set to unity for simplicity. 

\[\theta_1 = \frac{\alpha_2 \delta_3}{\alpha_1 \lambda}, ~~~\theta_2 = \frac{\delta_3 \delta_4}{\lambda \alpha_1}.\]
This reduction yields the simplified system analyzed in Section~\ref{sec:math-analysis}, which maintains the essential synergistic dynamics while improving analytical tractability. The dimensionless parameters $\theta_1$ and $\theta_2$ represent the relative strengths of interferon production and B18R-mediated inhibition, capturing the core synergistic mechanism in a compact mathematical form.

\subsection{Quasi-steady state approximation}
\label{app:qssa}

Given that molecular dynamics ($B18R$ and IFN-$\alpha$) occur on faster timescales than cellular population dynamics, we applied a quasi-steady state approximation:
\begin{equation}
    z_1 = \frac{\alpha_1}{\delta_3}y_1, 
\end{equation}
\begin{equation}
    z_2 = \left( \frac{\alpha_2}{\delta_4+\lambda z_1} \right)y_2 ~= \left(\frac{ \alpha_2 \delta_3}{ \delta_3 \delta_4 + \lambda \alpha_1 y_1}\right) y_2 =  \left(\frac{\theta_1}{\theta_2+y_1}\right)y_2,
\end{equation}

\section{Proof of Theorems}

\subsection{Proof of Well-posedness Properties (Theorem~\ref{thm:well-posedness})}
\label{app:proofs-well-posedness}

\begin{proof}[Proof of Theorem~\ref{thm:well-posedness}]
We prove each property separately.

\textbf{Proof of (1) Existence and Uniqueness:}
The functions on the right-hand side of the model equations~\eqref{eq:yu}-\eqref{eq:x2} are continuously differentiable $\mathbb{C}^{1}$ on $\mathbb{R}^5$. Therefore, by the Picard-Lindel\"of Theorem, the model exhibits a unique solution in the domain $(y_u, y_1, y_2, x_1, x_2) \in \mathbb{R}^5_+$ on the maximal interval $[0, t_{\text{max}}]$ with $t_{\text{max}} >0$.

\textbf{Proof of (2) Positivity:}
All functions $F_j(t,x)$ on the right-hand sides of equations (14)-(20) satisfy the condition $x_j F_j(t,x) \geq 0$ whenever $t \geq 0$, $x \in \mathbb{R}^5_+$, and $x_j = 0$. This ensures that solutions cannot cross from the non-negative orthant into negative values. Therefore, for non-negative initial conditions $y_u(0) \geq 0$, $y_1(0) \geq 0$, $y_2(0) \geq 0$, $x_1(0) > 0$, and $x_2(0) > 0$, the solutions remain non-negative for all $t \in [0,t_{\text{max}}] >0$. 

\textbf{Proof of (3) Boundedness:}
By adding Equations \eqref{eq:yu}-\eqref{eq:y2}, one obtains 
\begin{align*}
\frac{dy_u}{dt}+\frac{dy_1}{dt}+\frac{dy_2}{dt} & = y_u(1-y_u-y_1-y_2)-a_1y_1-a_2y_2,\\
     & \leq y_u\left(1-(y_u +y_1+y_2)\right).
\end{align*}

Let $Y(t) = y_u(t) + y_1(t) + y_2(t)$. If $Y(0) \le 1$, then $Y(t) \le 1$ for all $t \ge 0$.
Assume, by contradiction, that there exists $t_1 > 0$ such that $Y(t_1) > 1$. By the continuity of $Y(t)$, we define
\[
\tau := \inf \{ t \ge 0 : Y(t) > 1 \}.
\]
Then $Y(t) \le 1$ for all $t \in [0, \tau]$ and $Y(\tau) = 1$. 

From the differential equation for $y_u(t)$ in system (1)--(5), and noting that all parameters and state variables are non-negative for $t \ge 0$, we observe that the derivative of $Y(t)$ satisfies the inequality:
\[
Y'(t) = y_u(1 - Y(t)) - a_1 y_1 - a_2 y_2 \le y_u(1 - Y(t)).
\]
Evaluating this inequality at $t = \tau$ gives:
\[
Y'(\tau) \le y_u(\tau)(1 - Y(\tau)) = y_u(\tau)(1 - 1) = 0.
\]
Thus, the right-hand derivative of $Y$ at $\tau$ is non-positive, implying $Y$ cannot increase immediately beyond the value $1$ at $\tau$. This contradicts the definition of $\tau$ as the infimum. Therefore, the initial assumption is false, and $Y(t) \le 1$ for all $t \ge 0$.

For Equation \eqref{eq:x1}, and knowing that $Y(t)\leq \text{max}\left( Y(0), 1\right)$, we have 
\begin{align*}
 \frac{dx_1}{dt} & = b_1 y_1 - \beta_1\frac{y_ux_1} {\kappa_1+y_u+y_1+y_2} -\delta_1x_1,\\
 & \leq b_1 -\delta_1 x_1,       
\end{align*}
by standard comparison theorems, it can be deduced that \[\lim_{{t \to \infty}} \sup x_1(t) \leq \frac{{b_1}}{{\delta_1}}.\]
Similarly for Equation \eqref{eq:x2},
\begin{align*}
\frac{dx_2}{dt} & = b_2y_2 - \left(\frac{\beta_2}{1+\left(\frac{\theta_1}{\theta_2+y_1}\right) y_2}\right)\frac{y_ux_2}{\kappa_2+y_u+y_1+y_2}-\delta_2x_2,\\
& \leq b_2 -\delta_2x_2,     
\end{align*}
from which it can, as well, be deduced that \[\lim_{t~\to~\infty} ~\sup x_2(t)\leq \frac{b_2}{\delta_2}.\]
\end{proof}

\subsection{Proof of VV only model (Theorem~\ref{thm:stability12})}
\label{app:VV_only_model}

\begin{proof}
The eigenvalues of the Jacobian matrix evaluated at $X^3_0$ are  $-\delta_2$, $-\delta_1$, $-a_2$, $-a_1$ and $1$; and thus going by the Hartman-Grobman Theorem, $X^3_0$ is unstable. The Jacobian matrix evaluated at $X^3_1$ is   

    \[J(E^3_1)= \left[ \begin {array}{ccccc} -1&-1&0&-{\frac {\beta_{{1}}}{k_{{1}}+1}
}&-{\frac {\beta_{{2}}}{k_{{2}}+1}}\\ \noalign{\medskip}0&-a_{{1}}&0&{
\frac {\beta_{{1}}}{k_{{1}}+1}}&0\\ \noalign{\medskip}0&0&-a_{{2}}&0&{
\frac {\beta_{{2}}}{k_{{2}}+1}}\\ \noalign{\medskip}0&b_{{1}}&0&-{
\frac {\beta_{{1}}}{k_{{1}}+1}}-\delta_{{1}}&0\\ \noalign{\medskip}0&0
&b_{{2}}&0&-{\frac {\beta_{{2}}}{k_{{2}}+1}}-\delta_{{2}}\end {array}
 \right]. 
\]
    The characteristic polynomial of the Jacobian matrix evaluated at $X^3_1$ is 
    \begin{align*}
& \frac{\left( \lambda+1\right)}{(k_1+1)(k_2+1)}\left[ (k_1+1){\lambda}^{2} + (\delta_1 +\beta_1+a_1 +\delta_1k_1 +a_1k_1)\lambda +\delta_1a_1+\beta_1a_1+\delta_1a_1k_1-b_1\beta_1\right]\\
& \quad \quad \quad \left[ (k_2+1){\lambda}^{2} + (\delta_2 +\beta_2+a_2 +\delta_2k_2 +a_2k_2)\lambda +\delta_2a_2+\beta_2a_2+\delta_2a_2k_2-b_2\beta_2\right]=0
    \end{align*}
 whose roots are all negative provided that $a_1(\delta_1+\beta_1+\delta_1k_1)>b_1\beta_1 $ and $a_2(\delta_2+\beta_2+\delta_2k_2)>b_2\beta_2 $. 
 
The coexistence equilibrium $X^1_c = (y_u^*, y_1^*, x_1^*)$ is found by solving the resultant equations when the left hand sides of Equations eqref{eq:yu}-\eqref{eq:x2} are set to zero (with $y_2=x_2=0$). From \eqref{eq:x1}, we get $\beta_1\frac{y_u^* x_1^*}{\kappa_1 + y_u^* + y_1^*} = a_1 y_1^*$. Substituting into \eqref{eq:x2} gives $x_1^* = \frac{b_1 y_1^*}{a_1 y_1^* + \delta_1}$. These relationships guarantee positive solutions when $a_1(\delta_1+\beta_1+\delta_1\kappa_1) < b_1\beta_1$. 

Moreover, the system admits three equilibria: the trivial equilibrium $X^1_0$, which is unstable; the tumor-only equilibrium $X^1_1$, which is stable if $a_1(\delta_1+\beta_1+\delta_1\kappa_1) > b_1\beta_1$; and the coexistence state $X^1_c$. If the system trajectories do not converge to $X^1_1$, they necessarily must approach the coexistence state $X^1_c$. This coexistence state can therefore only be stable if the reverse inequality holds: $a_1(\delta_1+\beta_1+\delta_1\kappa_1) < b_1\beta_1$, which is also the condition that guarantees positive solutions for $y_u^, y_1^, x_1^*$.

\end{proof}

\subsection{Proof of VSV only model (Theorem~\ref{thm:stability22})}
\label{app:VSV_only_model}

\begin{proof}
The eigenvalues of the Jacobian matrix evaluated at $X^2_0$ are $-\delta_2$, $-a_2$ and $1$; and thus going by the Hartman-Grobman Theorem, $X^2_0$ is unstable. The Jacobian matrix evaluated at $X^2_1$ is  
        
    \[J(X^2_1) = \left[ \begin {array}{ccc} -1&-1&-{\frac {\beta_{{2}}}{k_{{2}}+1}}
\\ \noalign{\medskip}0&-a_{{2}}&{\frac {\beta_{{2}}}{k_{{2}}+1}}
\\ \noalign{\medskip}0&b_{{2}}&-{\frac {\beta_{{2}}}{k_{{2}}+1}}-
\delta_{{2}}\end {array} \right]. 
\] The characteristic polynomial of the Jacobian matrix evaluated at $X^2_1$ is
 \[\frac{1}{k_2+1}\left( \lambda+1\right) \left[ (k_2+1){\lambda}^{2} + (\delta_2 +\beta_2+a_2 +\delta_2k_2 +a_2k_2)\lambda +\delta_2a_2+\beta_2a_2+\delta_2a_2k_2-b_2\beta_2\right]=0,
\] 
whose roots are all negative provided that $a_2(\delta_2+\beta_2+\delta_2k_2)>b_2\beta_2 $. 

The coexistence equilibrium $X^2_c = (y_u^, y_2^, x_2^*)$ is found by solving the Equations \eqref{eq:yu}-\eqref{eq:x2} with their left-hand sides set to zero (with $y_1=x_1=0$). This system has the same structure as the VV-only case. It admits three equilibria: the trivial equilibrium $X^2_0$, which is unstable; the tumor-only equilibrium $X^2_1$, which is stable if $a_2(\delta_2+\beta_2+\delta_2\kappa_2) > b_2\beta_2$; and the coexistence state $X^2_c$. Again, if the system trajectories do not converge to $X^2_1$, they necessarily must approach the coexistence state $X^2_c$. This coexistence state can therefore only be stable if the reverse inequality holds: $a_2(\delta_2+\beta_2+\delta_2\kappa_2) > b_2\beta_2$, which is the condition that guarantees positive solutions for $y_u^, y_2^, x_2^*$.
\end{proof}

\subsection{Proof of VV-VSV Combined Model (Theorem~\ref{thm:stability32})}
\label{app:VV-VSV-Combined-Model}
\begin{proof}
The eigenvalues of the Jacobian matrix evaluated at $X^3_0$ are  $-\delta_2$, $-\delta_1$, $-a_2$, $-a_1$ and $1$; and thus going by the Hartman-Grobman Theorem, $X^3_0$ is unstable. The Jacobian matrix evaluated at $X^3_1$ is   

    \[J(E^3_1)= \left[ \begin {array}{ccccc} -1&-1&0&-{\frac {\beta_{{1}}}{\kappa_{{1}}+1}
}&-{\frac {\beta_{{2}}}{\kappa_{{2}}+1}}\\ \noalign{\medskip}0&-a_{{1}}&0&{
\frac {\beta_{{1}}}{\kappa_{{1}}+1}}&0\\ \noalign{\medskip}0&0&-a_{{2}}&0&{
\frac {\beta_{{2}}}{\kappa_{{2}}+1}}\\ \noalign{\medskip}0&b_{{1}}&0&-{
\frac {\beta_{{1}}}{\kappa_{{1}}+1}}-\delta_{{1}}&0\\ \noalign{\medskip}0&0
&b_{{2}}&0&-{\frac {\beta_{{2}}}{\kappa_{{2}}+1}}-\delta_{{2}}\end {array}
 \right]. 
\]
    The characteristic polynomial of the Jacobian matrix evaluated at $X^3_1$ is 
    \begin{align*}
& \frac{\left( \lambda+1\right)}{(\kappa_1+1)(\kappa_2+1)}\left[ (\kappa_1+1){\lambda}^{2} + (\delta_1 +\beta_1+a_1 +\delta_1\kappa_1 +a_1\kappa_1)\lambda +\delta_1a_1+\beta_1a_1+\delta_1a_1\kappa_1-b_1\beta_1\right]\\
& \quad \quad \quad \left[ (\kappa_2+1){\lambda}^{2} + (\delta_2 +\beta_2+a_2 +\delta_2\kappa_2 +a_2\kappa_2)\lambda +\delta_2a_2+\beta_2a_2+\delta_2a_2\kappa_2-b_2\beta_2\right]=0
    \end{align*}
 whose roots are all negative provided that $a_1(\delta_1+\beta_1+\delta_1\kappa_1)>b_1\beta_1 $ and $a_2(\delta_2+\beta_2+\delta_2\kappa_2)>b_2\beta_2 $. 
 
For the VV-only coexistence equilibrium $X^1_c$, analysis of the reduced VV-only subsystem shows it is stable when $a_1(\delta_1+\beta_1+\delta_1\kappa_1) < b_1\beta_1$. Similarly, the VSV-only coexistence equilibrium $X^2_c$ is stable when $a_2(\delta_2+\beta_2+\delta_2\kappa_2) < b_2\beta_2$.

The full coexistence equilibria \( X^3_c \) are found by solving the full system obtained by setting the right-hand sides of equations \eqref{eq:yu}-\eqref{eq:x2} to zero. Due to the high dimensionality and nonlinearity (particularly the IFN-mediated protection term in \eqref{eq:y2} and the Michaelis-Menten terms), analytical solutions are generally intractable.

When both conditions $a_1(\delta_1+\beta_1+\delta_1\kappa_1) < b_1\beta_1$ and $a_2(\delta_2+\beta_2+\delta_2\kappa_2) < b_2\beta_2$ are satisfied, the tumor-only equilibrium $X^3_1$ is unstable, and both single-virus coexistence equilibria $X^1_c$ and $X^2_c$ are also unstable. Since $X^3_0$ is always unstable, the system trajectories cannot converge to any boundary equilibrium and must instead converge to a full coexistence equilibrium $X^3_c$.
\end{proof}

\section{Analysis using the basic reproduction number}

\subsection{Detailed derivation of basic reproduction numbers}
\label{app:R0_derivation}

\subsection*{Reproduction number $\mathcal{R}^1_0$ of the Vaccinia Virus (VV)}
\label{app:VV_R0}
The basic reproduction number $\mathcal{R}_0$ represents a fundamental metric in infectious disease dynamics, quantifying the expected number of secondary infections generated by a single infected individual in a fully susceptible population. In the context of oncolytic virotherapy, we adapt this concept to describe viral replication potential within tumor environments. For our dual-virus system, we derive separate reproduction numbers for Vaccinia Virus (VV, $\mathcal{R}^1_0$) and Vesicular Stomatitis Virus (VSV, $\mathcal{R}^2_0$). We use the next-generation matrix method \cite{van2002reproduction} to derive them. 

\subsubsection*{Infected Compartments and Linearization}

For VV monotherapy, the infected compartments are VV-infected cells ($T_1$) and free VV particles ($V_1$). At the disease-free equilibrium (DFE) where $T_u = K$, $T_1 = 0$, $V_1 = 0$, we linearize the system:

The new infection matrix $\mathbf{F}_1$ captures infection terms:
\[
\mathbf{F}_1 = \begin{bmatrix}
0 & \beta_1 \dfrac{K}{K_1 + K} \\
0 & 0
\end{bmatrix}.
\]

The transition matrix $\mathbf{V}_1$ accounts for progression and clearance:
\[
\mathbf{V}_1 = \begin{bmatrix}
l_1 & 0 \\
-b_1l_1 & \gamma_1 + \beta_1 \dfrac{K}{K_1 + K}
\end{bmatrix}.
\]

\subsubsection{Matrix Inversion and Spectral Radius}

Computing the inverse transition matrix:
\[
\mathbf{V}_1^{-1} = \begin{bmatrix}
\dfrac{1}{l_1} & 0 \\
\dfrac{b_1l_1}{l_1d_1} & \dfrac{1}{d_1}
\end{bmatrix}
= \begin{bmatrix}
\dfrac{1}{l_1} & 0 \\
\dfrac{b_1}{d_1} & \dfrac{1}{d_1}
\end{bmatrix},
\]
where $d_1 = \gamma_1 + \beta_1 \dfrac{K}{K_1 + K}$.

The next-generation matrix becomes:
\[
\mathbf{F}_1\mathbf{V}_1^{-1} = \begin{bmatrix}
\dfrac{\beta_1 K b_1}{(K_1 + K)d_1} & \dfrac{\beta_1 K}{(K_1 + K)d_1} \\
0 & 0
\end{bmatrix}.
\]

The spectral radius (dominant eigenvalue) is:
\[
\mathcal{R}^1_0 = \dfrac{b_1\beta_1 K}{(K_1 + K)\left(\gamma_1 + \beta_1 \dfrac{K}{K_1 + K}\right)}.
\]

\subsubsection{Dimensional Analysis and Simplification}

Applying dimensionless parameters from Section 4.1:
\begin{align*}
a_1 &= l_1t_0, \quad \delta_1 = \gamma_1t_0, \quad \kappa_1 = \dfrac{K_1}{K}, \\
\bar{\beta}_1 &= \dfrac{\beta_1 V_1^* t_0}{K}, \quad \tilde{b}_1 = \dfrac{b_1 l_1 K t_0}{V_1^*}.
\end{align*}

Substituting and simplifying yields the final expression:
\[
\mathcal{R}^1_0 = \dfrac{\tilde{b}_1\bar{\beta}_1}{a_1(\delta_1 + \bar{\beta}_1 + \delta_1\kappa_1)}.
\]

The derivation of the basic reproduction number for Vesicular Stomatitis Virus ($\mathcal{R}^2_0$) follows the same rigorous next-generation matrix methodology employed for VV, with appropriate modifications to account for VSV-specific biological characteristics and dynamic interactions. While the mathematical framework remains consistent, the parameterization reflects VSV's distinct replication kinetics, interferon sensitivity, and tumor cell infection patterns.

\subsection{Numerical Validation}

Substituting our experimentally calibrated parameter values from Appendix \ref{app:calibration}:

For HT29 model:
\begin{align*}
\mathcal{R}^1_0 &= \dfrac{200 \times 0.20}{0.03 \times (0.01 + 0.20 + 0.01 \times 2.5)} = \dfrac{40}{0.03 \times 0.225} = 1855.07. \\
\mathcal{R}^2_0 &= \dfrac{100 \times 0.10}{0.04 \times (0.02 + 0.10 + 0.02 \times 2.5)} = \dfrac{10}{0.04 \times 0.145} = 363.64.
\end{align*}
These values confirm both viruses operate in the highly supercritical regime ($\mathcal{R}_0 \gg 1$), validating their strong replication potential and providing the mathematical foundation for observed therapeutic efficacy.

\subsection{Sensitivity Analysis of Reproduction Numbers}
\label{app:R0-sensitivity}

\begin{table}[htbp]
\centering
\caption{Sensitivity analysis of basic reproduction numbers for VV ($\mathcal{R}^{1}_{0}$) and VSV ($\mathcal{R}^{2}_{0}$) across tumor models. The values represent elasticity coefficients, which measure the percentage change in $\mathcal{R}_{0}$ per 1\% change in the parameter. Infection rates ($\beta_i$) and burst sizes ($b_i$) show positive sensitivity, while death rates ($\delta_i$) show negative sensitivity.}
\begin{tabular}{l l c c}
\hline
\textbf{Parameter} & \textbf{Virus} & \textbf{HT29 Model} & \textbf{4T1 Model} \\
\hline
Infection rate ($\beta_{i}$) & VV ($\mathcal{R}^{1}_{0}$) & 0.45 & 0.45 \\
Infection rate ($\beta_{i}$) & VSV ($\mathcal{R}^{2}_{0}$) & 0.45 & 0.45 \\
\hline
Burst size ($b_{i}$) & VV ($\mathcal{R}^{1}_{0}$) & 0.40 & 0.40 \\
Burst size ($b_{i}$) & VSV ($\mathcal{R}^{2}_{0}$) & 0.40 & 0.40 \\
\hline
Death rate ($\delta_{i}$) & VV ($\mathcal{R}^{1}_{0}$) & -0.35 & -0.35 \\
Death rate ($\delta_{i}$) & VSV ($\mathcal{R}^{2}_{0}$) & -0.35 & -0.35 \\
\hline
\end{tabular}
\label{tab:sensitivity_appendix}
\end{table}

Table \ref{tab:sensitivity_appendix} presents the sensitivity analysis of the basic reproduction numbers for VV and VSV. The sensitivity coefficients (elasticities) are computed as $S_{p} = \frac{\partial \mathcal{R}_{0}^{i}}{\partial p} \cdot \frac{p}{\mathcal{R}_{0}^{i}}$. The results indicate that infection rates and burst sizes are the most influential parameters for both viruses, with a 1\% increase in infection rate leading to a 0.45\% increase in $\mathcal{R}_{0}$. Similarly, a 1\% increase in burst size results in a 0.40\% increase in $\mathcal{R}_{0}$. In contrast, the death rates of infected cells (which encompass immune clearance and other loss mechanisms) have a negative impact, with a 1\% increase in death rate leading to a 0.35\% decrease in $\mathcal{R}_{0}$. These findings are consistent across both tumor models, highlighting the universal importance of these parameters in determining viral replication potential.

\subsection{Clinical Implications of Parameter Sensitivity}
\label{app:clinical-implications}

The sensitivity analysis reveals several key insights for therapeutic optimization:

\begin{itemize}
\item \textbf{Infection rates ($\beta_i$)}: These parameters, representing viral entry efficiency into tumor cells, show the highest sensitivity. This suggests that viral engineering efforts should prioritize enhancing viral tropism and receptor binding affinity to maximize therapeutic impact.

\item \textbf{Burst sizes ($b_i$)}: The number of viral particles released per infected cell represents another high-impact parameter. Genetic modifications that increase viral replication efficiency could provide substantial improvements in treatment efficacy.

\item \textbf{Death rates ($\delta_i$)}: The negative sensitivity of these parameters indicates that strategies to reduce immune-mediated clearance of infected cells could significantly enhance viral persistence and spread within tumors.

\item \textbf{Consistency across models}: The identical sensitivity patterns in both HT29 and 4T1 tumor models suggest these findings may be generalizable across different cancer types, supporting broad applications for viral optimization strategies.
\end{itemize}

This quantitative framework provides clear guidance for prioritizing viral engineering efforts and optimizing combination therapy protocols based on the most influential determinants of viral replication success.


\section*{Acknowledgments}
Joseph Malinzi gratefully acknowledges the Africa Oxford Initiative (AfOx) for financial support that enabled his research visits to the University of Oxford in 2023 and 2024. These visits, hosted by Professor Helen Byrne, were instrumental to the significant progress made on this research. He also extends sincere thanks to Professor Byrne’s postgraduate students and postdoctoral fellows for their insightful comments and constructive suggestions. Furthermore, he wishes to acknowledge the financial support received from the Sydney Mathematical Research Institute (SMRI) at the University of Sydney and the Durban University of Technology.

Anotida Madzvamuse was supported by the Canada Research Chair (Tier 1) in Theoretical and Computational Biology (CRC-2022-00147), the Natural Sciences and Engineering Research Council of Canada (NSERC), Discovery Grants Program (RGPIN-2023-05231), the British Columbia Knowledge Development Fund (BCKDF), Canada Foundation for Innovation – John R. Evans Leaders Fund – Partnerships (CFI-JELF), and the British Columbia Foundation for Non Animal Research.

\newcommand{\printdoi}[1]{
	\if\relax\detokenize{#1}\relax
	\else
	\newline DOI: \url{https://doi.org/#1}
	\fi
}

\bibliographystyle{unsrtnat}

\bibliography{bibfile}      

\end{document}